\documentclass[preprint]{elsarticle}
\pdfoutput=1
\journal{ArXiv Preprints}

\usepackage{pgfplots}\usetikzlibrary{plotmarks}
\pgfplotsset{compat=newest}
\usepackage{csquotes}
\usepackage{subfig}
\usepackage{url}
\usepackage[justification=centering]{caption}
\usepackage[norelsize,ruled, lined]{algorithm2e}
\usepackage{amsmath}
\usepackage{amssymb}
\usepackage{graphicx}
\usepackage{float}
\graphicspath{ {images/} }

\newtheorem{theorem}{Theorem}
\newtheorem{lemma}[theorem]{Lemma}
\newtheorem{proposition}[theorem]{Proposition}
\newdefinition{remark}{Remark}
\newproof{proof}{Proof}
\newdefinition{definition}{Definition}
\newdefinition{problem}[definition]{Problem}
\begin{document}

\begin{frontmatter}

\title{The Minimum Edge Compact Spanner Network Design Problem}

\author[rvt]{Tathagata Mukherjee\corref{cor1} }
\ead{tathagata@intelligentrobotics.org}
\author[focal]{Alexander Veremyev }
\ead{averemyev@gmail.com}
\author[els]{Piyush Kumar}
\ead{piyush@compgeom.com}
\author[elp]{Eduardo Pasiliao Jr}
\ead{elpasiliao@gmail.com}
\cortext[cor1]{Corresponding author}
\address[rvt]{Intelligent Robotics Inc.}
\address[focal]{University of Central Florida}
\address[els]{Compgeom Inc.}
\address[elp]{Air Force research Labs.}

%\author{Tathagata Mukherjee\fnref{myfootnote1}}
%\address{Intelligent Robotics Inc.}
%\fntext[myfootnote1]{tathagata@intelligentrobotics.org}

%\author{Alexander Veremyev\fnref{myfootnote2}}
%\address{University of Central Florida}
%\fntext[myfootnote1]{tathagata@intelligentrobotics.org}




\begin{abstract}
In this paper we introduce and study the Minimum Edge Compact Spanner~(MECS) problem. We prove hardness results
related to the problem, design exact and greedy
algorithms for solving the problem, and show related experimental
results.
% In this paper, we introduce the the Minimum Edge Compact Spanner~(MECS) Problem, prove related hardness results, design exact and greedy algorithms for the problem, and
% show experimental results.
% Network design deals with problems on 
% weighted and unweighted graphs. 
The MECS problem looks for sparse subgraphs of an input graph, such that the average shortest path distance 
is preserved to a constant factor. 
Average distance is a measure of the ease of communication over the network.
%, represented by the graph. 
As a 
result such problems have applications in areas where one wants to substitute a dense graph with a sparse subgraph while maintaining 
a low cost of communication. 

% More precisely, we choose the factor to be a constant multiple of the average distance of 
% the input graph. Average distance is a measure of the ease of communication over the network, represented by the graph. As a 
% result such problems have applications in areas where one wants to substitute a dense graph with a sparse subgraph while maintaining 
% a low cost of communication.
\end{abstract}


\end{frontmatter}

%\linenumbers

\section{Introduction}
%!TEX root=main.tex

% spell check has been done on this document
% reading complete

In this paper we define and study the minimum edge 
compact spanner~(MECS) problem.
The MECS problem is based on the idea
of average distance or average path length~(APL) in a directed or undirected graph. 

The average distance, also called the average all-pairs-shortest-path distance or APL, 
denoted by $\mu$, is defined as the average of the shortest path distance between all the 
vertex pairs.  This metric can be used to measure efficiency of information
flow over a network~[\cite{ellens2013graph}]. The meaning of the
APL as a measure of robustness of a communication network, follows from
the fact that the shorter the APL, the more robust the network. As a result, the
APL is one of the three most robust measures of network topology, 
along with its clustering coefficient and degree distribution~[\cite{ellens2013graph}]. Examples of 
the use of APL in a network, include the average number of clicks
used to go from one website to another and the average number of hops
that one might go through, to get in touch with a complete stranger in a social
network. 

The APL has important implications in network design. In a real 
network, like the World Wide Web, a short APL facilitates the quick transfer 
of information from one node to another and hence reduces costs. The efficiency of mass 
transfer in a metabolic network can be judged by studying its APL~[\cite{jeong2000large}]. A power 
grid network will have less losses if its APL is minimized. Most real networks 
have a very short APL leading to the concept of a small world, where everyone is 
connected to everyone else through a very short path. As a result, most models of real networks 
are created with this condition in mind. One of the first models which tried to explain real 
networks was the random network model~[\cite{kirkby1976tests}]. It was later followed by the 
models of Watts and Strogatz~[\cite{watts1998collective}] and the random graph model of networks
~[\cite{newman2002random}]. Still later there were scale-free networks starting with the BA model 
~[\cite{adamic2000power}]. All these models had one thing in common: they all predicted networks with 
very short average path length~(APL). 

The average path length is different from the \emph{diameter} of a graph. 
The \emph{diameter} is defined as the \emph{longest} shortest path between
any two nodes in a graph~[\cite{bondy1976graph}]. It is easy to see that the APL is bounded 
above by the diameter. However, in most cases it is much shorter than the diameter.
Thus the APL can be used to measure the average performance of the 
network whereas the diameter is used to measure the worst case behavior. Moreover,
the APL is an upper bound to the independence number of a 
graph~[\cite{chung1988average}].

%Given a graph $G$ the APL can change in one of two possible ways: first if the vertex set of the graph changes, 
%then so does the APL and second if the edges of the graph change, then so does the APL. 

Given a graph $G$, the APL can change if the vertex set or the edges of 
the graph change. 
%In this work we study the change in the APL induced by the second scenario. 
%We assume that the edge set changes, because of \emph{edge deletions}. 
In this work, we assume that the edge set changes, because of \emph{edge deletions}. 
Such deletions
have the potential to increase the APL of the graph. Our goal is to
find subgraphs of the original graph, obtained through the deletion of edges, such that
the APL does not increase too much with respect to the APL of the original graph. 
In essence, we are looking at the following problem:
given a graph $G$, we want to find a sparse subgraph $G_s$, such that the average
path length of $G_s$ is bounded above by a constant. We choose this constant to be a
constant multiple of the average path length of the original graph $G$.

A related problem that has been studied extensively is that of a spanner~
[\cite{althofer1993sparse}]. Spanners are sparse subgraphs of a graph $G$,
such that the shortest path distance between every pair of nodes is preserved
up to a given factor. Next we formally define a spanner but before moving forward
we note that given a graph $G=(V,E)$, we can define a metric space using the all pairs
shortest path distances in the graph. Thus we define a spanner using the underlying 
metric space. Formally, a spanner is defined as follows:

\begin{definition}[Spanner] \label{spannerdef} 
Let $(V, d)$ be a finite metric space. An undirected graph $G=(V,E)$ is a
\emph{$t$-spanner} for $V$ if for every pair of vertices $x,y \in V$ we have that 
$d_{G(x,y)} \leq t \cdot d(x,y)$, where $d_{G(x,y)}$ is the length of the shortest path from 
$x$ to $y$ in $G$ (where the length of an edge $\{u,v\} \in E$ is $d(u,v)$). $t$ is called the
\emph{stretch} of the spanner.
\end{definition}

We note that a spanner preserves the following, to a constant factor, given by the stretch $t$:
(1) the shortest path distance between every vertex pair is preserved to a constant factor
(2) the diameter of the graph is preserved to a constant factor and (3) the APL is preserved 
to a constant factor. Thus if we want to preserve the APL of a graph while sparsifying it,
we can use any $t$-spanner algorithm for doing the same.

However, for a communication network with a small APL, we can ensure communication efficiency of a 
subnetwork by bounding its APL by a small constant. We do not need to ensure that all the shortest 
path distances or the diameter are preserved. Therefore using standard $t$-spanner algorithms for 
computing such sparse subgraphs is an overreach. This gives us the motivation 
to define a new problem, that we call the \emph{Minimum Edge Compact Spanner problem}. 
The goal of this problem is to find sparse subgraphs of the original input graph such that their APL is bounded 
above by a small constant. In general this constant can be any design parameter of the underlying network. For our 
purposes, we consider it to be a multiple of the APL of the original network.

Informally, the MECS problem asks the following question: given a graph $G$, is there a \emph{sparse} subgraph 
$G_s$ such that the APL is preserved up to a constant factor. Formally, let $G=(V,E,W)$ be a graph with vertex 
set $V=\{1, \dots, n\}$ and edge set $E\subset V\times V$ and a weight function $W:V \times V \rightarrow \mathbb{R}$. 
Although we consider the weights to be real valued, the  MIP implementations are not able to handle 
arbitrary weights on the edges of the graphs. Moreover, we prove that the underlying problem is $NP$-complete in the simpler
case when the underlying weights are integers, which in turn proves the $NP$-hardness of the original problem. 
To formalize the settings, let the average distance in the graph $G$ be denoted by $\mu_G$. For the weighted graph we 
define $\mu_G$ as follows: 

\begin{definition}
Let $v_i,v_j \in V$. Let $d(v_i,v_j)$ denote the weight of the shortest path
between the vertices $v_i$ and $v_j$. Then the average distance is defined as
$\mu_G = \frac{\sum_{i \neq j} d(v_i,v_j)}{n(n-1)}$, where $n = |V|$
\end{definition}

Let $\mu_{G_s} \geq \mu_G$ be the average distance of the sparse subgraph $G_s$.
Now the problem can be stated as follows: given a positive constant $t > 1$ find a subset 
$E_s \subseteq E$ with the \emph{minimum} number of edges~(or minimum total edge weight) 
such that the average distance $\mu_{G_s}$ of the sub-graph $G_s=(V,E_s,W_s)$, where 
$W_s$ is the restriction of $W$ to $G_s$, satisfies $\mu_{G_s} \leq t\mu_G$ or more generally 
$\mu_{G_s} \leq C$ for some constant $C$. Thus the basic optimization problem can be posed as 
the following integer program:
\begin{align*}
\textrm{minimize} \sum_{e \in E} x_e \cdot W_e \\ \textrm{subject to} \ \mu_{G_s} \leq t \mu_G \\ x_e \in \{ 0,1 \}
\end{align*}
The variable $x_e$ is an indicator variable and takes a value of $x_e=1$ if $e \in E_s$ and $x_e=0$
otherwise. We note that the condition $\mu_{G_s} \leq t \mu_G$ encodes the requirement that the resulting
subgraph needs to be connected. For if not then this condition is violated as in that case
$\mu_{G_s} = \infty$. We also note that, more generally we can pose the following problem:
\begin{align*}
\textrm{minimize} \sum_{e \in E} x_e \cdot W_e \\ \textrm{subject to} \ \mu_{G_s} \leq C \\ x_e \in \{ 0,1 \}
\end{align*}
$C$ is a constant that can depend on the average path length of the input graph, either
in an additive or a multiplicative way. Formally we pose the \emph{MECS} problem as follows:

\begin{definition}[Minimum Edge Compact Spanner (MECS) Problem]
Given a graph $G=(V,E,W)$ with the average distance $\mu_G$ and a positive constant $t > 1$ find a 
subset $E_s \subseteq E$ with the minimum number of edges such that the average
distance $\mu_{G_s}$ of the graph $G_s=(V,E_s,W_s)$ satisfies $\mu_{G_s} \leq C$ where $C$ is a 
constant or we may set $C = t\mu_G$ or $C = \mu_G + \delta$ for some small $\delta > 0$.
We also note that $W_s$ is a restriction of $W:V \times V \rightarrow \mathbb{R}$ to $G_s$.
\end{definition}

As mentioned before, we are the first to introduce the MECS problem. Going forward we first review some
of the work on related problems and then we establish the complexity of the MECS problem. More precisely, 
we show that the MECS problem is $NP$-hard. Then we study exact algorithms for the MECS problem based on
integer programs. Finally we study greedy algorithms for this problem and compare the performance of the exact
algorithms with them.

\section{Previous Work}
%!TEX root=main.tex

% spell check has been done on this document
% finished complete read 

The problem of sparse spanners for general graphs was introduced in~[\cite{peleg1989graph,peleg1987optimal}]. It was proved that the 
problem is $NP$-complete and since then it has been studied extensively. Several different variations of the sparse spanner problem 
has also been studied and several results on the complexity of such problems is also available. For example, one such variation is 
the minimum stretch spanning tree problem. In [\cite{cai1995tree}] it was proved that for a given graph $G$, the problem of 
deciding whether $G$ has a tree $t$-spanner is $NP$-complete for any fixed $t \geq 4$ and is linear time solvable for $t = 1, 2$. For 
the problem of \emph{tree} $t$-spanners, an $O(\log n)$ approximation algorithm is also known for finding the smallest value of $t$ 
for which such a spanner exists~[\cite{emek2004approximating}]. The problem of approximability of sparse spanners for general graphs
was studied in~[\cite{elkin2000strong,elkin2000hardness,elkin2005approximating}].

Among the several different variations of the sparse spanner problem, the ones that we are interested in, relate to the reliability 
of the underlying sparse spanner. One variation, that has been widely studied, is that of a \emph{Fault tolerant spanner}. 
These were first studied by~[\cite{levcopoulos2002improved}]. The first results on fault tolerant spanners for 
general graphs was given in~[\cite{chechik2010fault}], which were later improved by~[\cite{dinitz2011fault}]. 
Subsequently fault tolerant spanners have been studied extensively~[\cite{solomon2012fault,solomon2014hierarchical,
chan2012sparse}] and several different constructions for fault tolerant spanners are known in Euclidean as well as 
doubling spaces~[\cite{solomon2012fault,dinitz2011fault}]. As in the case of a simple $t$-spanner, we are interested 
in finding subgraphs that satisfy the $t$-spanner property. However, we have the added constraint that the resulting
subgraph should be tolerant to vertex faults, that is the subgraph should still be a $t$-spanner even if a certain predetermined
number of vertices fail. Thus if the resulting subgraph is tolerant to $k$ vertex failures, it is called the $k$-fault
tolerant $t$-spanner. Another variation of spanners is called the \emph{robust spanner}~[\cite{bose2013robust}]. Here if a subset 
$S \subseteq V$ of vertices fail, then a super set $S^* \supseteq S$ are affected and the rest of the graph with $V \backslash S^*$ 
vertices is still a $t$-spanner.

A related class of problems is that of \emph{Network Design} and have been studied for a long time. The sparse spanner problem is 
a special instance of a network design problem. The general description of a network design problem is as follows: given a set 
of nodes (offices, switches), possible links, costs for each link, and either the number of permitted link or node failures between 
each pair of nodes, we want to design a cost effective communication network. Some of the well known network design 
problems are the minimum spanning tree problem~[\cite{cormen2009introduction}], the Steiner tree problem~[\cite{hwang1992steiner}],
the survival network design problem~(SNDP)~[\cite{botton2013benders}], the uniform buy-at-bulk network design problem
~[\cite{awerbuch1997buy}], and the traveling salesman problem~[\cite{cormen2009introduction}]. Most of these problems
attempt to find subgraphs of an input graph with various \emph{spanning} properties~[\cite{ma2016minimum,golden2005heuristic,melkonian2005primal}]. 
For example, in the Steiner tree problem, we are given a graph with vertex set $V$ and a set of terminal nodes $T$. 
The goal is to find a minimum cost tree in the input graph that connects all the terminals. In the survival network design 
problem~(SNDP) we are given, for every pair of vertices $i,j$ in the input graph, an integer $r_{ij}$. The goal is to find 
a subgraph in which there are at least $r_{ij}$ edge disjoint paths between the vertices $i$ and $j$. This is very similar 
to the problem of a finding a fault tolerant spanner, the difference being that in the case of fault tolerant spanners we 
consider vertex faults whereas in the case of SNDP we consider path faults. Interested readers may refer to~[\cite{gupta2011approximation}] 
for an excellent survey on approximation algorithms for network design problems.

One of the earliest results on the hardness of network design problems was presented in~[\cite{johnson1978complexity}]. 
They proved that the problem of finding a subgraph with minimum APL under weight constraints is $NP$-complete. A study on the worst case behavior 
of heuristics for this problem was studied by Wong in~[\cite{wong1980worst}]. 
Subsequently many variations of network design problems have been studied in the literature. Most of the network design problems are 
known to be $NP$-complete. A comprehensive survey of approximation algorithms for the variants of the network design problem can 
be found in~[\cite{gupta2011approximation,chuzhoy2008approximability}]. To the best of our knowledge the Minimum Edge Compact 
Spanner~(MECS) network design problem has never been studied before. There are two problems that are close to the one that we study.
The first is that of the $\epsilon$-slack spanner~[\cite{chan2006spanners}]. This paper studied the problem of \emph{spanners with slack}. 
They are defined as follows:

\begin{definition}[$\epsilon$-slack spanner] 
Suppose that $M=(V,d)$ is a metric. Then $H_{\epsilon}=(V,E_H)$ is an $\epsilon -$slack spanner for the metric $M=(V,d)$ 
if for given $0 < \epsilon < 1$, for any vertex $v \in V$ the furthest $(1-\epsilon) \cdot n$ vertices $w$, have the property 
that $d(v,w) \leq d_H (v,w) \leq D \cdot d(v,w)$ for some constant $D$. There is no constraint on the distances to the nearest
$\epsilon \cdot n$ vertices, where $n = |V|$
\end{definition}
Note that $d_H$ is the shortest path metric on the graph $H_{\epsilon}$. The resulting spanner uses the furthest $(1-\epsilon) \cdot n$ vertices. 
The $\epsilon \cdot n$ vertices that are nearby are not required to satisfy the constraint on the path length. 
In the paper Chan et al. gave a construction for these structures and showed that they distort the average path length by
a constant factor. We note from the definition of the problem, that this is not the MECS problem. In fact the MECS
problem is a more general version of the $\epsilon$-slack spanner problem. On the other hand the problem studied in [\cite{johnson1978complexity}]
looks for subgraphs that minimize the APL subject to constraints on the weight of the
underlying graph. 

The most common approach for distance-based network design problems which takes into account the distances between nodes in 
networks is the flow-based method~[\cite{Botton11,Gouveia08,Pirkul03}]. Other methods based on the MIP based approach include the 
path based approach~[\cite{veremyev_cnp_2015}]. However, one of the major problems of these exact methods for solving $NP$-complete
problems is the time it takes to solve problems of moderate sizes. As mentioned before the MECS problem has not been studied 
before and as such we do not know of any work that studies this problem from the perspective of exact algorithms.

\section{Hardness.}
%!TEX root=main.tex
% spell check has been done on this document
% All todos have been taken care of in this document
% a complete reading has also been done

Now we are ready to state and prove hardness results for the MECS problem. More precisely, we prove that the MECS problem
is $NP$-complete. We prove NP-completeness in the simpler case when the edge weights are integers. From the perspective 
of complexity, it makes sense to just show that even if we restrict ourselves to the simpler case where the edge weights of the 
underlying graph are integers, the problem is still $NP$-complete.
We assume that we are given a graph $G=(V,E,W)$ where $W$ is a function from $E$ to $\mathbb{N}$.
Thus, given an edge $e \in E$ we have an integer weight $W(e)$ associated with $e$. As before we denote 
the average path length~(APL) in the graph $G$ as $\mu_G$ and the MECS problem aims to find a subgraph 
$G_s \subseteq G$  such that $\mu_{G_s} \leq t \cdot \mu_G$ for some $t \geq 1$ 
and has the minimum total weight among all such graphs.  We note that for $t=1$ we return the graph $G$ as the solution. 
So we consider the problem for $t > 1$. We also note that the solution to the problem has to be a connected graph. 
For if the graph is not connected the average distance, by definition, is infinite.

Our proof derives from the construction used by [\cite{johnson1978complexity}] for proving the $NP$-completeness 
of the classical network design problem where the goal is to minimize the average path length of a graph subject 
to constraints on its weight. Thus the problem considered in [\cite{johnson1978complexity}] is a complementary one 
and hence the underlying decision problems take the same form. This gives
us the opportunity to use their constructions for proving the $NP$-completeness of the MECS problem as well.
In some sense we reinvented the constructions and then realized that they have been used before. As a result
some of the technicalities of the proofs are a bit different and hence we have decided to present them in
detail here.

In order to prove that the MECS problem is $NP$-hard we consider the decision version of the problem and 
show that it is $NP$-complete. Let $t > 1$ and let us write $t \cdot \mu_G = c$. Let the spanning subgraph 
under the constraint of APL be denoted by $G_s$. Then the decision version of the MECS problem is as follows: 
Does there exist a subgraph $G_s=(V,E_s,W_s)$, $E_s \subseteq E$ and $W_s$ is a restriction of $W$ to $G_s$ 
such that:

\begin{equation}
\sum_{e \in E_s} W_e \leq r \textrm{ and} \ \mu_{G_s} \leq c, \ c \ \textrm{finite}
\end{equation}

Instead of writing $\mu_{G_s}$ in the above equation, we could simply replace it with the sum of the
all pairs shortest path weights. This will have the effect of changing the constant $c$ to $C$, where
$c = \frac{C}{n(n-1)}$ where $n=|V|$.

\begin{equation} \label{desmecs} 
\sum_{e \in E_s} W_e \leq r \textrm{ and} \sum_{(u,v):u,v \in V;u \neq v} d(u,v) \leq C, \ C \ \textrm{finite}
\end{equation}

It is easy to see that the problem is in $NP$, for given a graph it is easy to check whether it is a feasible solution of
problem [\ref{desmecs}]. Now we proceed to prove that this problem is $NP$-complete. 
We establish the $NP$-completeness in two ways. The first one uses a reduction from the Subset Sum~[\cite{cormen2009introduction}] 
problem. This is a standard reduction and the proof is short and concise. However, it still leaves the possibility 
of having special and simple instances of the problem which are not $NP$-hard. This is where our second proof 
comes in. The second proof shows that the problem is $NP$-hard even if we restrict ourselves to finding spanning trees
having the MECS property. This is a much more stronger result  and goes on to show that the problem is hard
even in the simplest of all cases, namely when we are looking for spanning trees. We start with the reduction from 
\emph{Subset Sum}~[\cite{cormen2009introduction}], which is a known $NP$-complete problem. We state the subset sum problem 
below:

\begin{definition}[Subset Sum Problem] \label{ss}
Given a set $U$ of integers $\{a_1,a_2, \ldots, a_k \}$ and a target integer $b$, the subset sum problem asks the 
following question: is there a subset $S$ of $\{ 1,2, \ldots, k \}$ such that $\sum_{i \in S} a_i = b$ 
\end{definition}

This problem [SP13] from [\cite{garey1979computers}] is known to be $NP$-complete.

\subsection{NP Completeness} 
 
In order to prove $NP$-completeness we use a gadget that given an instance of the subset sum problem creates an instance of 
the MECS feasibility problem~[\ref{desmecs}]. Then we show that solving this new problem is equivalent to solving the original
subset sum problem.

\begin{definition}[Construction for Reduction] \label{npc}
Consider an instance of the subset sum problem with the universe $U = \{a_1,a_2, \ldots, a_k \}$.
We note that $k= O(1)$ and is an input to the problem. We create a graph $G$ using the following steps.

\begin{enumerate} 
\item For every $i \in \{ 1,2, \ldots, k \}$ create two nodes $i,i'$ in $G$ 
\item Add a node $N$ to the graph such that the vertex set for $G$ is $N \cup \{i,i'\}, i \in \{1,2, \ldots, k \}$ 
\item The edge set $E$ consists of the following edges: for every $i \in \{1,2, \ldots, k \}$ add the edges $(i,i'),(N,i),(N,i')$ 
\item The weight function is defined as follows: $W(i,i')=a_i$, $W(N,i')=a_i$, $W(N,i)=a_i$ 
\item Let $T = \sum_{i \in \{1,2, \ldots, k \}} a_i$. We set $r = 2T+b$ and $C = 4kT-b$ in problem [\ref{desmecs}]
\end{enumerate}

\end{definition}

The resulting graph, which is an input to problem [\ref{desmecs}] is shown in figure [\ref{fig:input}].

%\begin{figure}[t]
%\centering
%\includegraphics[width=12cm]{spoke_g.pdf}
%\caption{The graph G used for the reduction}
%\centering
%\label{fig:graph}
%\end{figure}

%\begin{figure}[t]
%\centering
%\includegraphics[width=12cm]{triangle_g.pdf}
%\caption{The spanning tree of G used for the reduction}
%\centering
%\label{fig:st_graph}
%\end{figure}

%\begin{figure}[t]
%\centering
%\captionsetup{justification=centering}
%\begin{minipage}{.5\textwidth}
%  \centering
%  \includegraphics[width=.8\linewidth]{spoke_g.pdf}
  %\captionof{figure}{The Spoke Graph}
%  \label{fig:spoke}
%\end{minipage}%
%\begin{minipage}{.5\textwidth}
%  \centering
%  \includegraphics[width=.8\linewidth]{triangle_g.pdf}
  %\captionof{figure}{A Feasible Solution}
%  \label{fig:input}
%\end{minipage}
%\caption{Gadgets for NP-hardness: The first is the spoke graph and the second is the input to the NDP solver}
%\end{figure}

\begin{figure}[!tbp]
  \centering
  \subfloat[]{\includegraphics[width=0.4\textwidth]{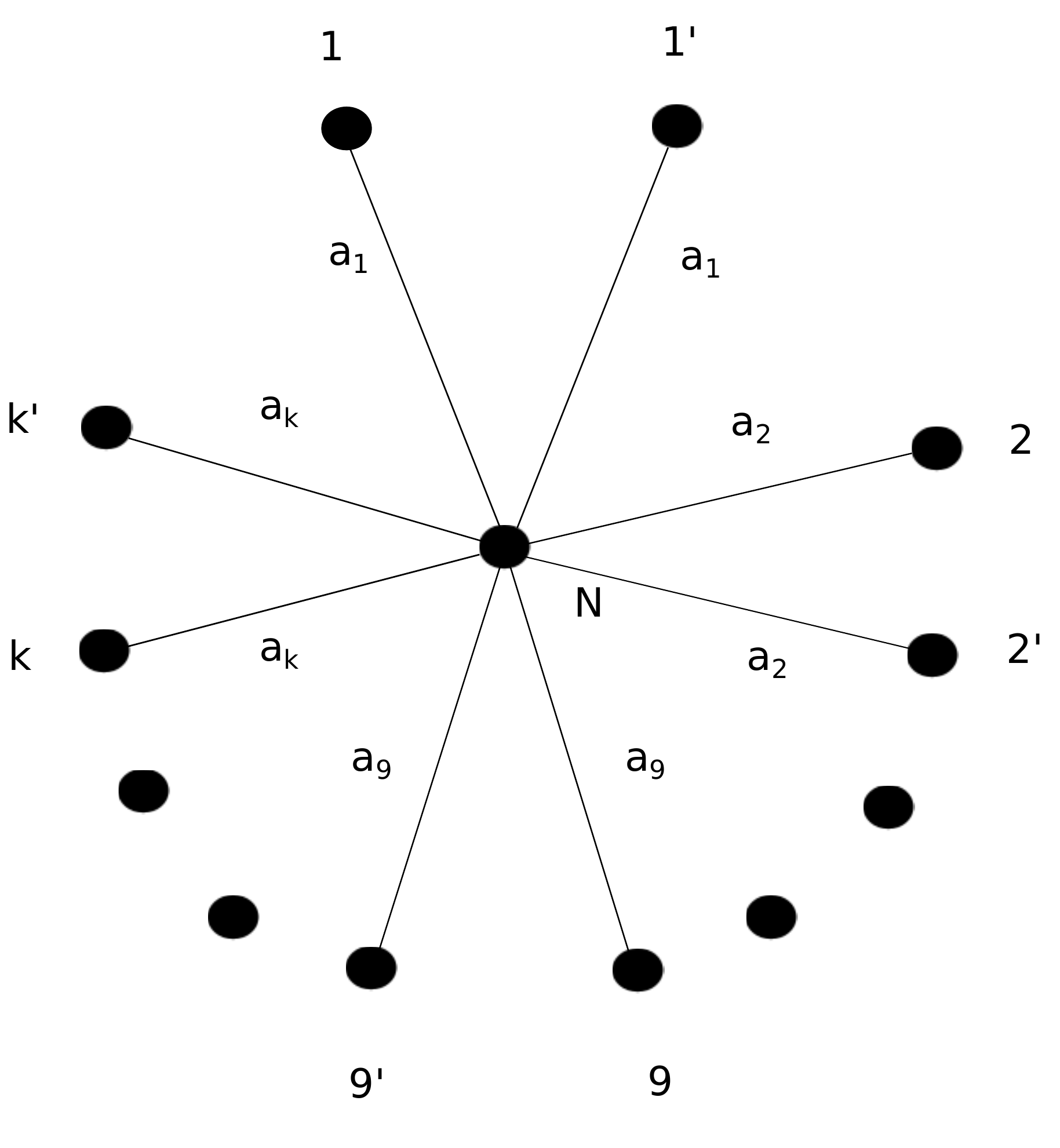}\label{fig:spoke}}
  \hfill
  \subfloat[]{\includegraphics[width=0.4\textwidth]{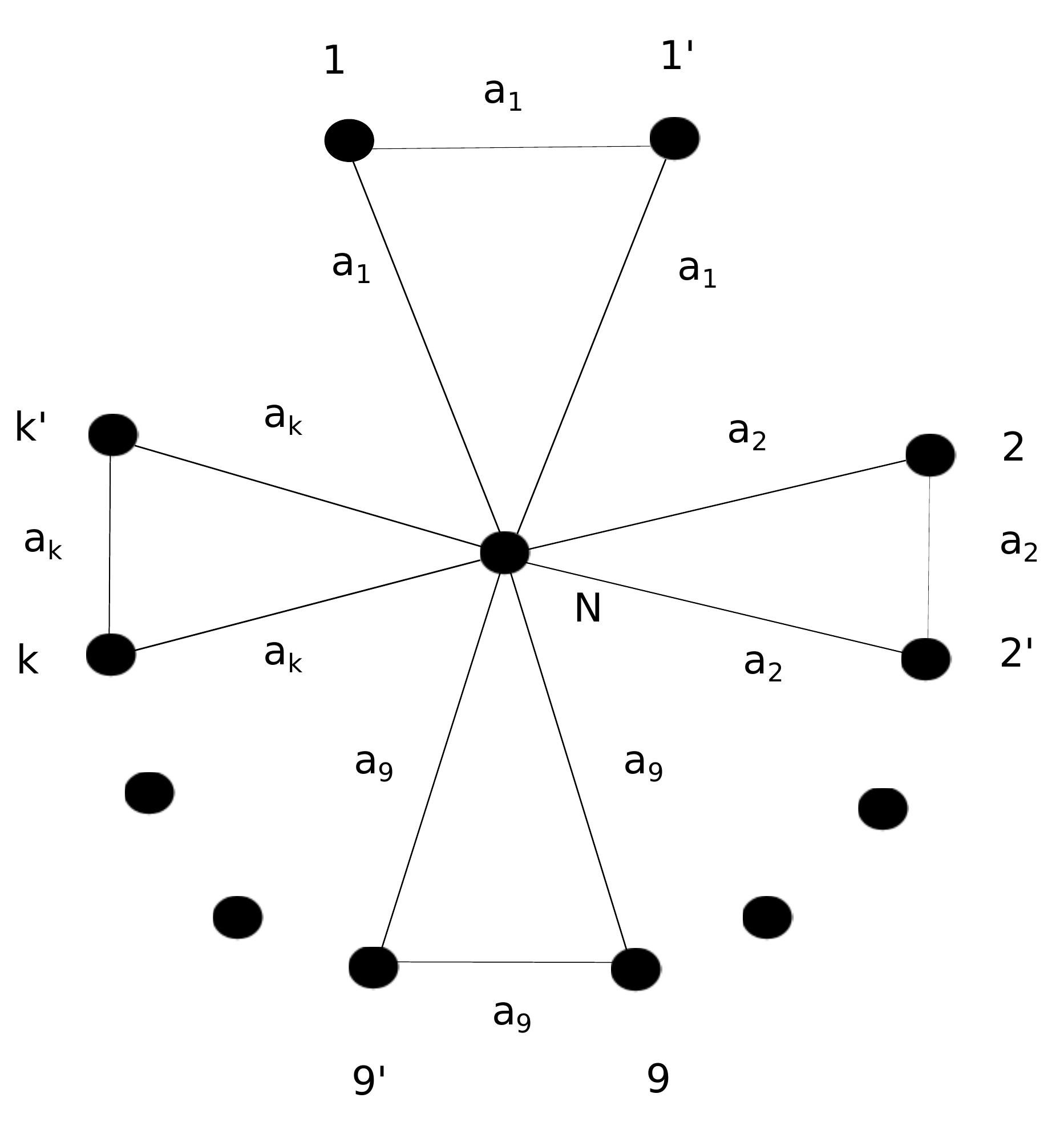}\label{fig:input}}
  \caption{Illustration of the graphs used for the $NP$-completeness proof of the MECS problem (a) The graph that create
  a \emph{spoke} pattern and is included in every feasible solution, this is the graph $G_f$ (b) The input to the MECS feasibility solver}
\end{figure}

We claim that if we are given a feasible solution for problem [\ref{desmecs}] on $G$, then we can get a solution for 
problem [\ref{ss}] in polynomial time. In order to do this we define another graph $G_f$ as follows: $G_f=(V,E_f,W_f)$ 
where $E_f = \{(N,i),(N,i') \},i \in \{1,2,\ldots, k\}$ as before and $W_f$ contains the corresponding weights
Now we are ready to prove the $NP$-completeness of the MECS problem. We establish this through the following lemmas. 

\begin{lemma}
The total edge weight of $G_f$ is $2T$.
\end{lemma}

\begin{proof}
The proof is straightforward. The total weight is given by $\sum_{i \in \{1,2,3,\ldots k \} } (W(N,i) + W(N,i'))$. Using
the fact that $W(N,i) = a_i$, and the fact that $T = \sum_{i \in \{1,2, \ldots, k \}} a_i$, the result follows.
\end{proof}

The next lemma looks at the sum of the all pairs shortest path weights in the graph $G_f$.

\begin{lemma}
The sum of the all pairs shortest paths weight of $G_f$ is $4kT$.
\end{lemma}

\begin{proof}
We prove this using induction. When $k=1$ the result holds, as in that case we have 3 pairs of nodes and the sum of the shortest
path weights is $W(N,1) + W(N,1') + W(SP(1,1')) = a_1 + a_1 + 2 \cdot a_1 = 4 \cdot a_1$. Let us suppose that the result holds for 
$k=m$. Now we prove that the result holds for $k = m+1$. Thus we assume that when $k=m$, the sum of the all pairs shortest paths 
weight is $4\cdot m \cdot \sum_{i \in \{1,2,\ldots,m \} } a_i$. Now suppose that $k=m+1$. This adds two more nodes and two more edges
to the existing graph. This in turn means that with every pair of nodes $(i,i'); i \in \{ 1,2,\ldots,m \} $, another four pairs of
shortest path weights are added to the existing total. The weights of these four pairs is $4 \cdot \left(a_i + a_{m+1} \right); i \in \{1,2,\ldots,m \} $.
Moreover, the total weight of the shortest path distances between the newly added nodes is $4 \cdot a_{m+1}$. Thus the addition
of the two new nodes \emph{increases}, the sum of the all pairs shortest path weights by:
$$4 \cdot \sum_{i \in \{1,2,\ldots,m \}} \left(a_i + a_{m+1}\right) + \left(4 \cdot a_{m+1} \right) $$ Thus the new sum of the all pairs 
shortest path weights is:
%TODO: T: Can you please fix the following sum. Make sure you try \left( and \right)
$$ \left(4\cdot m \right) \cdot \left( \sum_{i \in \{1,2,\ldots,m \} } a_i \right) + 4 \cdot \left(\sum_{i \in \{1,2,\ldots,m \}} a_i \right) +  
4 \cdot \left(m+1 \right) \cdot a_{m+1} $$ 
The above sum comes out to be $4\cdot (m+1) \cdot \sum_{i \in \{1,2,\ldots,m+1 \} } a_i$. Thus the proof follows by induction. 
\end{proof}

In the next lemma we investigate the result of adding the edge $(i,i'); i \in \{ 1,2,\ldots,m \}$ to the graph $G_f$.

\begin{lemma} \label{npl2} 
Let $G_f=(V,E_f,W_f)$ where $E_f = \{(N,i),(N,i') \},i \in \{1,2,\ldots, k\}$ as before and $W_f$ contains the 
corresponding weights. The addition of the edge $(i,i')$ to $G_f$ increases the sum of edge weights 
by $a_i$ and decreases the sum of the all pairs shortest path weight by $a_i$.
\end{lemma}

\begin{proof}
The proof follows from the fact that the total edge weight of $G_f$ is $2T$ and the total all pairs
shortest paths weight of $G_f$ is $4kT$. Its easy to see that if the edge $(i,i')$ is added to $G_f$, then the total
weight increases by $a_i$ (by construction \ref{npc}) and the total shortest path weight decreases by $a_i$ because of
the fact that the addition of the edge between $(i,i')$ changes the shortest path between $(i,i')$ and decreases the weight 
of the shortest path from $2 \cdot a_i$ to $a_i$. 
\end{proof}

The next lemma establishes the fact that any feasible solution of problem [\ref{desmecs}] can be assumed to contain $G_f$ as 
a subgraph.

\begin{figure}[t]
\centering
\includegraphics[width=.5\linewidth]{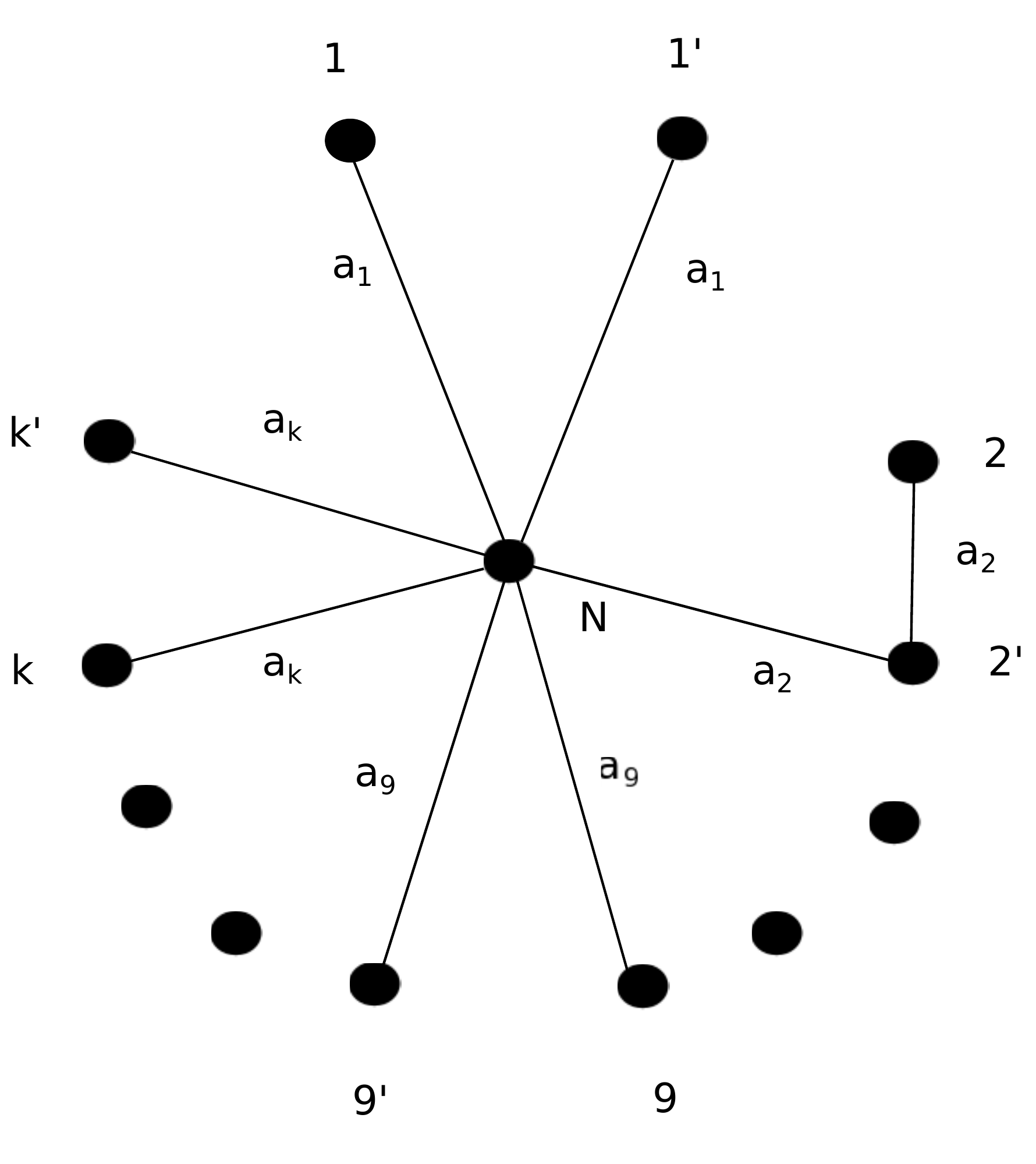}
\caption{Example graph used for lemma \ref{npl1}. The graph is almost similar to the spoke graph that we have used
before, but one of the spokes, namely that for vertex 2 is broken and instead we have a pattern like the number 7 in its
place}
\centering
\label{fig:ex_graph}
\end{figure}

\begin{lemma} \label{npl1} 
Define $G_f=(V,E_f,W_f)$ where $E_f = \{(N,i),(N,i') \},i \in \{ 1,2,\ldots, k\}$ and $W_f$ contains the corresponding weights. 
Any feasible solution of the problem [\ref{desmecs}] on $G$ can be assumed to contain $G_f$ as a subgraph. 
\end{lemma} 

\begin{proof}
%We first note that the sum of the edge weights of $G_f$ is $2T$ and the sum of the shortest path weights 
%is $4kT$. Moreover, addition of the edge $(i,i')$ to $G_f$ increases the sum of edge weights by $a_i$ 
%and decreases the sum of the all pairs shortest path weight by $a_i$, by lemma [\ref{npl2}].  
%Then it must contain the graph $G_{f'} = (V_{f'},E_{f'},W_{f'})$ where $V_{f'}= V$ and 
%$E_{f'}=\{(N,i),(i,i') \},i \in \{ 1,2, \ldots, k \}$
%and $W_{f'}$ contains the corresponding weights. For if not then the resulting graph is disconnected and hence cannot 
%be feasible. Moreover, for the graph $G_{f'}$ the sum of the edge weights is $2T$. However, the sum of the shortest path
%weights is not equal to $4kT$.
Let us suppose that the feasible solution to the problem [\ref{desmecs}] does not contain the graph $G_f$.
Then the feasible solution would contain a graph that is different from $G_f$. However the bad news is that 
there are several different possibilities for such graphs. The good news is that all these possibilities have
one thing in common, each unit consisting of the three vertices $i,i',N$ can have only two possible structures
apart from the one where $i$ and $i'$ are both connected to $N$. In both these settings $i$ is connected to $i'$ and 
either $i$ is connected to $N$ or $i'$ is connected to $N$. Thus a feasible solution can have an underlying graph where
some of the triplets are arranged to form a spoke like structure and the others are arranged in one of the other two 
possible ways.

In order to see what happens in this setting we consider the simplest variation of the graph shown in figure [\ref{fig:spoke}].
The graph is such that there is a vertex pair $i, i'$, such that the edges associated with this vertex pair are $(N,i),(i,i')$ and the rest
of the graph is a spoke graph. An example is shown in the figure [\ref{fig:ex_graph}]. It is easy to see that for the 
graph in figure [\ref{fig:ex_graph}], the sum of the edge weights is $2T$. However, the sum of weights of the all pairs shortest paths
is $4kT + (k-1) (2 \cdot a_2)$. Now suppose that a feasible solution $F$ contains this graph as a subgraph. As the solution is feasible
we have $ \sum_{e \in E_s} W_e \leq 2T+b $ and $\sum_{(u,v):u,v \in V;u \neq v} d(u,v) \leq 4kT -b $. Now, if the feasible solution contains 
[\ref{fig:ex_graph}] as a subgraph, then in order to satisfy the constraint on the all pairs shortest path weights, one would need to add
$(i,i')$ edges~(each of which reduces the sum of the shortest path weights by $a_i$) such that the sum of the all pairs shortest path
weights becomes less than or equal to $4kT-b$. This in turn means that the sum of the edge weights will increase by at least the same
amount, namely $b + (k-1) (2 \cdot a_2)$. Thus the sum of the edge weights will be at least $2T + b + (k-1) (2 \cdot a_2)$.
But this will violate the first constraint thus making the solution infeasible. Thus this creates a contradiction. 
In particular, more the number of vertices $i$ with edges of the form $(N,i),(i,i')$, the larger
is the sum of the all pairs shortest path weights. Thus we can use the same argument as above to disprove the fact that any feasible
solution will create such a subgraph. Hence we can assume that any feasible solution contains the graph $G_f$ and the 
proof follows. 

\end{proof}

\begin{theorem}
The MECS problem is $NP-$complete. 
\end{theorem} 

\begin{proof}
Let us consider a feasible solution of the problem [\ref{desmecs}] such that equality holds for both the inequalities. 
Let $G_s$ be the graph corresponding to this solution. $G_s$ is a subgraph of $G$ and $G_f$ is a subgraph of $G_s$ by lemma [\ref{npl1}]. 
Consider the edge set $E_s$ of $G_s$ and consider exactly the edges $(i,i')$. For every such edge, set 
$S= S \cup \{ i\}$. Then it follows that $\sum_{i \in S} a_i =b$ by the feasibility of the solution and lemma [\ref{npl2}] 
and thus we have got a solution of the problem [\ref{ss}]. This completes the proof.
\end{proof}

Next we prove a stronger result that precludes the possibility of finding a spanning tree that satisfies the constraint 
on the APL, in polynomial time. In particular we show that the MECS problem is $NP$-complete even in the case, 
where we restrict ourselves to spanning trees of unweighted graphs. Here the problem is as follows: Given an undirected 
graph $G=(V,E)$ with average distance $\mu_G$ and a finite real number $c$, find a subgraph $G_s=(V,E_s)$ such that, 

\begin{equation} \label{ecsts}
|E_s| \leq |V|-1 \ \textrm{ and} \ \mu_{G_s} \leq c, \ c \ \textrm{finite}
\end{equation} 

We call this the \emph{Edge Compact Spanning Tree Spanner~(ECSTS) Problem}. We prove that the ECSTS problem is $NP$-complete.
The actual reduction for the proof is shown in Appendix A because of space constraints.

\section{Exact Algorithms}
\label{formulation}
%!TEX root=main.tex
% spell check has been done on this document

Going forward we study solutions to the MECS problem by formulating the problem as a mathematical program. More precisely we look at two approaches: 
the first approach is based on the idea of flows as described in~[\cite{Botton11,Gouveia08,Pirkul03}] and the second approach is based upon formulating the
problem using a mixed integer program. It must be noted that, going forward, whenever we mention spanners we mean a solution to the MECS problem,
which is based on the notion of the average path length and not spanners as referred to in the standard spanner literature.

\subsection{Flow-based Approach}
The most common approach for distance-based network design problems, which take into account the distances between nodes in networks, 
is the flow-based method [\cite{Botton11,Gouveia08,Pirkul03}]. Namely, let $f^{st}_{ij}\in \{ 0,1\}$ for all $s,t\in V$ and $(i,j)\in E$ 
denote the flow sent from vertex $s\in V$ to vertex $t\in V$ through an edge $(i,j)\in E$, $s<t$. Assuming that the total amount of flow 
sent from vertex $s\in V$ to vertex $t\in V$ through a spanner $G_s=(V,E_s)$ is 1, one can deduce that the length of the path that the flow 
takes from $s$ to $t$, in a spanner $G_s=(V,E_s)$ is $ \sum\limits_{(i,j)\in E} f^{st}_{ij} $; hence, the average path-length of all flows 
in a spanner is:
$$
\mu_f=\frac{2}{n(n-1)}\sum\limits_{s,t\in V:i<j}\sum\limits_{(i,j)\in E_s} f^{st}_{ij}
$$
Note that $\mu_f\geq \mu_s$; however, there always exists a flow $\mathbf{f}=\{f^{st}_{ij}|s,t\in V, (i,j)\in E\}$ such that 
$\mu_f= \mu_s$ (a flow which uses only the shortest paths). Then, the problem formulation can be written as follows: 

\begin{problem}[flow-based MECS]
\begin{subequations}
\begin{align}
&\mbox{minimize} \sum_{(i,j)\in E}x_{ij} \label{spanner_flow}\\
&\quad\mbox{subject to} \nonumber\\
&\frac{2}{n(n-1)}\sum\limits_{s,t\in V:i<j}\sum\limits_{(i,j)\in E_s} f^{st}_{ij} \leq t\mu_G &\label{spanner_flow_main}\\
&\sum\limits_{j:(s,j) \in E}  f^{st}_{sj}-\sum\limits_{i:(i,s) \in E}  f^{st}_{is} \geq 1 & \forall s,t\in V,\ s<t,\label{spanner_flow_1}\\
&\sum\limits_{i:(i,t) \in E}  f^{st}_{it}-\sum\limits_{j:(t,j) \in E}  f^{st}_{tj}\geq 1 & \forall s,t\in V,\ s<t,\label{spanner_flow_2}\\
&\sum\limits_{j: (i,j) \in E}  \left(f^{st}_{ij}-f^{st}_{ji}\right)=0 & \forall s,t\in V, \ s<t, \ \forall i\in V\setminus\{s,t\},\label{spanner_flow_3}\\
& f^{st}_{ij}\leq   x_{ij} &\forall s,t\in V,\ s<t, \ \forall (i,j) \in E, \label{spanner_flow_4}\\
& x_{ij} \in \{0,1\}, \ 0\leq  f^{st}_{ij}\leq 1 &\forall s,t\in V, \  s<t,\ \forall (i,j)\in E.
\end{align}
\end{subequations}
\end{problem}

In the formulation above, constraint \eqref{spanner_flow_main} is the main constraint on average path-length in a spanner, and 
constraints \eqref{spanner_flow_1}-\eqref{spanner_flow_4} are the standard flow-balancing constraints. Note that we relax the 
binary requirement on variables $f^{st}_{ij}$. It is easy to verify that the formulation is still correct. This flow-based formulation 
requires $O(|V|^2|E|)$ variables and constraints.

\subsection{Path-based Approach}
In this section, we develop a path-based approach, which is similar to the one presented in [\cite{veremyev_cnp_2015}] and is based on 
introducing new distance-based variables to compute the average path-length. The main idea is to define path variables for each pair 
of nodes and enforce constraints on them recursively. We demonstrate that such an approach allows to reduce the number of variables and 
constraints. 

Let $u^{(\ell)}_{ij}$ be a binary variable such that $u^{(\ell)}_{ij}=1$ if and only if there is a path of length at most $\ell$ between 
nodes $i$ and $j$ in a spanner $G_s$, where $i,j\in V$, $\ell \leq L$ and $L\leq |V|-1$ is an appropriate constant. Let also   
$u^{(0)}_{ij}=0$ and $u^{(L+1)}_{ij}=1$  for simplicity. In addition, we define $y^{(\ell)}_{ikj}$ to be a binary variable such that   
$y^{(\ell)}_{ikj}=1$ if there is a path between nodes $i$ and $j$ of length at most $\ell$ in a spanner $G_s$ which traverses a neighbor 
$k\neq j$ of node $i$ where $i,j\in V$, $(i,k)\in E$ and $\ell \leq L$.

Then, the formulation can be written as
\begin{problem}[path-based MECS]
\label{P3_opt}
\begin{subequations}
\begin{align}
&\mbox{minimize} \sum_{(i,j)\in E} x_{ij} \label{f1_obj}\\
&\mbox{subject to} \nonumber\\
& \sum\limits_{\ell=1}^{L+1} \sum\limits_{i,j=1:i<j}^n \ell(u^{(\ell)}_{ij}-u^{(\ell-1)}_{ij})\leq t\mu_G \frac{n(n-1)}{2} \label{f1_main}\\
& u^{(1)}_{ij}= x_{ij}, &\hspace{-5mm} (i,j)\in E\label{f1-edge_1}\\
& u^{(1)}_{ij}= 0, &\hspace{-5mm} (i,j)\notin E\label{f1-edge_2}\\
& u^{(\ell)}_{ij}\geq u^{(\ell-1)}_{ij}, &\hspace{-5mm}\forall i,j\in V,\ \ell\in\{2,\dots,L\}\label{f1-w_1}\\
& u^{(\ell)}_{ij}\leq x_{ij}+\sum\limits_{k\neq j:(i,k)\in E}y^{(\ell)}_{ikj}, &\hspace{-5mm}\forall (i,j)\in E,\ \ell\in\{2,\dots,L\}\label{f1-w_21}\\
& u^{(\ell)}_{ij}\geq \frac{1}{deg_i}\left(x_{ij}+\sum\limits_{k\neq j:(i,k)\in E}y^{(\ell)}_{ikj}\right), &\hspace{-5mm}\forall (i,j)\in E,\ \ell\in\{2,\dots,L\}\label{f1-w_31}\\
& u^{(\ell)}_{ij}\leq \sum\limits_{k\neq j:(i,k)\in E}y^{(\ell)}_{ikj}, &\hspace{-5mm}\forall (i,j)\notin E,\ \ell\in\{2,\dots,L\}\label{f1-w_22}\\
& u^{(\ell)}_{ij}\geq \frac{1}{deg_i}\left(\sum\limits_{k\neq j:(i,k)\in E}y^{(\ell)}_{ikj}\right), &\hspace{-5mm}\forall (i,j)\notin E,\ \ell\in\{2,\dots,L\}\label{f1-w_32}\\
& y^{(\ell)}_{ikj}\leq x_{ik},\ y^{(\ell)}_{ikj}\leq u^{(\ell-1)}_{kj},  &\hspace{-25mm}i,j\in V, (i,k)\in E, \ell\in\{2,\dots,L\}\label{f1-y_1}\\
& y^{(\ell)}_{ikj}\geq x_{ik} + u^{(\ell-1)}_{kj} -1, &\hspace{-25mm} i,j\in V, (i,k)\in E, \ell\in\{2,\dots,L\}\label{f1-y_2}\\
&x_{ij},\ u^{(\ell)}_{ij},\ y^{(\ell)}_{ikj} \in \{0,1\}, &\hspace{-5mm}\forall i,j,k,\in V,\ \ell\in\{1,\dots,L\}.\label{f1-end}
\end{align}
\end{subequations}
\end{problem}

In the formulation above constraint \eqref{f1_main} is the main constraint on average distance. Note that if the shortest path length between a pair of nodes $i,j\in V$ is $d$, then $u^{(d)}_{ij}-u^{(d-1)}_{ij}=1$ and $u^{(\ell)}_{ij}-u^{(\ell-1)}_{ij}=0$  for $d\neq \ell$. Constraints \eqref{f1-edge_1} - \eqref{f1-w_32} recursively model paths variables $u^{(\ell)}_{ij}$. Constraints \eqref{f1-y_1} - \eqref{f1-y_2} recursively model additional paths variables  $y^{(\ell)}_{ikj}$. The formulation requires $|E|$ variables $x_{ij}$, $L|V|(|V|-1)/2$ variables $u^{(\ell)}_{ij}$, and $L|V||E|$ variables $y^{(\ell)}_{ikj}$.

%Explanations
%\begin{itemize}
%\item Constraint \eqref{f1_main} is the main constraint on average distance. Note that if the shortest path length between a pair of nodes $i,j\in V$ is $d$, then $u^{(d)}_{ij}-u^{(d-1)}_{ij}=1$ and $u^{(\ell)}_{ij}-u^{(\ell-1)}_{ij}=0$  for $d\neq \ell$. We can substitute $\mu_s=t\mu$ if we use a starching factor $t$
%\item Constraints \eqref{f1-w_1} - \eqref{f1-w_3} recursively model paths variables
%\item Constraints \eqref{f1-y_1} - \eqref{f1-y_2} recursively model extra paths variables
%\end{itemize}
%Number of variables and constraints
%\begin{itemize}
%\item $|E|$ - $x_{ij}$ variables, in the formulation above constraints where $x_{ij}$ do not exist need to be separated
%\item $Ln(n-1)/2$ - $w^{(\ell)}_{ij}$ variables
%\item $Ln|E|$ - $y^{(\ell )}_{ikj}$ since just neighbors of node $i$ need to be considered for all $i$
%\item Note that if the distance between $i,j$ is greater than $L$, it is counted as $L+1$
%\end{itemize}

Note that if the distance between $i,j$ is greater than $L$, it is counted as $L+1$ by the formulation above.
Hence, to appropriately compute the average distance, we need to use $L=n-1$ as the maximum possible spanner diameter. However, in practice one can expect the spanner diameter to be $~ln(n)$ as many real-world and randomly graph topologies exhibit a so-called ``small-world'' property [\cite{watts1998collective,albert2002statistical}]; therefore, in Section \ref{MIP_algorithm} we develop an exact iterative MIP-based algorithm to solve the problem more efficiently.

\subsection{Formulation enhancements}
Below, we outline various formulation enhancements which may help to improve the solvers performance.

\subsubsection{Redundant constraints and integrality relaxation}
Observe that for any $i<j\in V$
\begin{align}
\sum\limits_{\ell=1}^{L+1}  \ell(u^{(\ell)}_{ij}-u^{(\ell-1)}_{ij}) &= u^{(1)}_{ij}+2(u^{(2)}_{ij}-u^{(1)}_{ij})+\dots+ L(u^{(L)}_{ij}-u^{(L-1)}_{ij})+(L+1)(1-u^{(L)}_{ij})=\nonumber\\
&=L+1-\sum\limits_{\ell=1}^{L} u^{(\ell)}_{ij}
\end{align}
therefore making $w^{(\ell)}_{ij}$ as large as possible will still make the problem feasible and satisfy the main constraint \eqref{f1_main} on the spanner average distance. Therefore, only constraints on upper bound on $u^{(\ell)}_{ij}$ and $y^{(\ell)}_{ikj}$ are needed. Note that these constraints are enforced only to make the aforementioned variables equal to zero, hence, these variables do not need to be integral anymore. Only variables $x_{ij}$ need to be binary, and we need $|E|$ of them.

%\begin{problem}[Minimum-edge compact spanner (Optimized)]
%\label{P3_opt}
%\begin{subequations}
%\begin{align}
%&\mbox{minimize} \sum_{(i,j)\in E} x_{ij} \label{f1o_obj}\\
%&\mbox{subject to} \nonumber\\
%& \sum\limits_{\ell=1}^{L+1} \sum\limits_{i,j=1:i<j}^n \ell(w^{(\ell)}_{ij}-w^{(\ell-1)}_{ij})\leq \mu_s \frac{n(n-1)}{2} \label{f1o_main}\\
%& w^{(1)}_{ij}= x_{ij},\  x_{ij}\leq \mathbbm{1}_{(i,j)\in E} &i,j\in V\label{f1o-edge}\\
%& w^{(\ell)}_{ij}\geq w^{(\ell-1)}_{ij} &\forall i,j\in V,\ \ell\in\{2,\dots,L\}\label{f1o-w_1}\\
%& w^{(\ell)}_{ij}\leq x_{ij}+\sum\limits_{k:(i,k)\in E}y^{(\ell)}_{ikj} &\forall i,j\in V,\ \ell\in\{2,\dots,L\}\label{f1o-w_2}\\
%%& w^{(l)}_{ij}\geq \frac{1}{deg_i}\left(x_{ij}+\sum\limits_{k\neq j:(i,k)\in E}y^{(l)}_{ikj}\right) &\forall i,j\in V,\ \ell\in\{2,\dots,L\}\label{f1o-w_3}\\
%& y^{(\ell)}_{ikj}\leq x_{ik},\ y^{(\ell)}_{ikj}\leq w^{(\ell-1)}_{kj}  &i,j\in V, (i,k)\in E, \ell\in\{2,\dots,L\}\label{f1o-y_1}\\
%%& y^{(l)}_{ikj}\geq x_{ik} + w^{(\ell-1)}_{kj} -1 &i,j\in V, (i,k)\in E, \ell\in\{2,\dots,L\}\label{f1o-y_2}\\
%&x_{ij}, \in \{0,1\}, &\forall i,j\in V\label{f1o-end1}\\
%&w^{(\ell)}_{ij},\ y^{(\ell)}_{ikj} \in [0,1], &\forall i,j,k,\in V,\ \ell\in\{1,\dots,L\}.\label{f1o-end}
%\end{align}
%\end{subequations}
%\end{problem}

\subsubsection{Leaf-node considerations}

Denote by $N_1$ all nodes of $V$ with degree 1 in $G$, i.e., $N_1=\{i\in V\ | \ deg_G(i)=1\}$. Such nodes are also referred to as leaves (or leaf nodes) of $G$. Let $n_1=|N_1|$. Since a spanner cannot have isolated nodes, then the edges connecting leaf node to the network must be included into any spanner. 

Thus, one should simply enforce $x_{ij}=1$ for all $i\in N_1, (i,j)\in E$. However, the computational experiments show that it is more beneficial of not considering leaf nodes and edges going to leaf nodes in the formulation. In this case, the average distance of a spanner can be computed by the following expression:

\begin{subequations}
\begin{align}
\mu=&\frac{1}{n(n-1)}\Bigg\{ \sum_{i,j\in V\setminus N_1:\ i<j}\left( u^{(1)}_{ij}+\sum\limits_{\ell=2}^{L+1} \ell \left(u^{(\ell)}_{ij}-u^{(\ell-1)}_{ij}\right)\right)\label{vcc:begin-gen0}\\
&+\sum_{i\in V\setminus N_1}v_i +\sum_{i\in V\setminus N_1}2\times \frac{v_i(v_i-1)}{2} \label{vcc:begin-gen2}\\
&+\sum_{i,j\in V\setminus N_1:\ i<j}(v_i+v_j)\left( 2u^{(1)}_{ij}+\sum\limits_{\ell=2}^{L} (\ell+1) \left(u^{(\ell)}_{ij}-u^{(\ell-1)}_{ij}\right)\right)\label{vcc:begin-gen3}\\
&+\sum_{i,j\in V\setminus N_1:\ i<j}(v_i\cdot  v_j)\left( 3u^{(1)}_{ij}+\sum\limits_{\ell=2}^{L-1} (\ell+2) \left(u^{(\ell)}_{ij}-u^{(\ell-1)}_{ij}\right)\right)\Bigg\}\label{vcc:begin-gen4},
\end{align}
\end{subequations}
\noindent which explicitly represents the fact that nodes and edges from $N_1$ are not considered. Pre-computed parameter $v_i$, $i\in V\setminus N_1$, is equal to the number of neighbors of  $i\in V\setminus N_1$ that belong to $N_1$, i.e., $v_i=|\mathcal{N}_G(i)\cap N_1|$, where $\mathcal{N}_G(i)=\{j:(i,j)\in E\}$ is a set of neighbors of node $i$. The terms in (\ref{vcc:begin-gen2}) correspond to the paths (of length 1) from $i\in N\setminus N_1$ to $\mathcal{N}_G(i)\cap N_1$, and paths (of length 2) between nodes in $\mathcal{N}_G(i)\cap N_1$, respectively. Next, for any $i,j\in N\setminus N_1$, $i<j$, the term in (\ref{vcc:begin-gen3}) represent all paths between $i$ and $\mathcal{N}_G(j)\cap N_1$ as well as $j$ and $\mathcal{N}_G(i)\cap N_1$. Similarly, (\ref{vcc:begin-gen4}) computes all paths between $\mathcal{N}_G(i)\cap N_1$ and $\mathcal{N}_G(j)\cap N_1$.  

\subsubsection{Other inequalities}
\begin{itemize}
\item {\bf Isolated nodes consideration.} Since the spanner cannot have isolated nodes, the following inequality can be used to enforce that
\begin{equation}
\sum\limits_{j:(i,j)\in E} x_{ij}\geq 1, \quad \forall i\in V.
\end{equation}
\item {\bf Connectivity consideration.} Since the spanner has to be connected, then it should have at least $n-1$ edges
\begin{equation}
\sum\limits_{(i,j)\in E} x_{ij}\geq n-1, \quad \forall i\in V.
\end{equation}
\item {\bf Connectivity Violating Cuts.} Let $\mathcal{E}_c$  be a set of edge cuts, i.e., subsets of edges $E_\alpha\subset E$, such that for any $E_\alpha \in \mathcal{E}_c$  a graph $G_\alpha=(V,E\setminus E_\alpha)$ is disconnected. A spanner should have at least one edge from each edge cut $E_\alpha \in \mathcal{E}_c$:
\begin{equation}
\sum\limits_{(i,j)\in E_\alpha} x_{ij}\geq 1, \quad \forall E_\alpha\in \mathcal{E}_c
\end{equation}
\item {\bf Average Distance Constraint Violating Cuts.} Let $\mathcal{E}_d$  be a set of subsets of edges $E_\beta\subset E$, such that 
for any $E_\beta \in \mathcal{E}_d$ the average distance of a graph $G_\beta=(V,E\setminus E_\beta)$ is greater than $t\mu_G$. 
A spanner should have at least one edge from each $E_\beta \in \mathcal{E}_d$:
\begin{equation}
\sum\limits_{(i,j)\in E_\beta} x_{ij}\geq 1, \quad \forall E_\beta\in \mathcal{E}_d
\end{equation}
\end{itemize}

\subsection{Exact MIP-based algorithm}
\label{MIP_algorithm}
The number of variables and constraints used in Problem \ref{P3_opt} is $O(L|V||E|)$, hence the lower the value of $L$, the better 
solver performance we can expect. In this section, we develop an exact MIP-based algorithms, which is based on the following proposition. 

\begin{proposition}
Let $E^*_s$ be the optimal solution of Problem \ref{P3_opt} for $L=L_0$ and $G_s=(V,E_s)$. If $diam(G_s)\leq L_0$, then $G_s$ is also 
an optimal spanner, i.e.,  $E^*_s$ is the optimal solution of Problem \ref{P3_opt} for $L=n-1$.
\end{proposition}

\begin{proof}
For any spanner $G_s=(V,E_s)$ let $\lfloor \mu_s\rfloor_L= \frac{1}{n(n-1)}\sum\limits_{i,j=1}^n \min(d_{G_s}(i,j),L)$,  which can be 
viewed as a truncated version of the average distance $\mu_s$, where any distance of greater than $L$ in a graph $G_s$ is treated as $L$. 
In fact, the constraint \eqref{f1_main} of Problem \ref{P3_opt} for any given $L$ restricts the truncated average distance of a spanner 
$\lfloor \mu_s\rfloor_L$.

%the truncated average distance of an optimal spanner  $E^*_s$ obtained as a solution of \ref{P3_opt} for a given $L$ is not greater than $\lfloor \mu_s\rfloor_L$
Note that $\lfloor \mu_s\rfloor_L\leq \mu_s$ for any $L=1,\dots, n-1$. Hence, any feasible solution of Problem \ref{P3_opt} with $L=n-1$ 
is also a feasible solution for Problem \ref{P3_opt} for any  $L=1,\dots,n-1$. Moreover, if $L\geq  diam(G_s)$, then $\lfloor \mu_s\rfloor_L= \mu_s$. 
Therefore, the optimal spanner $E^*_s$ is also a feasible solution of Problem \ref{P3_opt} with $L=n-1$, and, hence, is also optimal.

\end{proof}

Hence, this problem can be solved sequentially. First, we set $L=diam(G)$ and solve the corresponding MIP. If the diameter of the obtained 
spanner greater than $L$, we set $L=L+1$ and solve the problem again until the diameter of the spanner will not be greater than $L$. This 
technique was implemented in  [\cite{veremyev_cnp_2015}] and according to their observations, it allows to substantially reduce the computational 
time and the amount of variables. 

Below is the formal algorithm description.
\begin{algorithm}
%\LinesNumbered
\DontPrintSemicolon
%$\KwData{A graph $G=(N,E)$, metric $f()$, and budget $C$}
%\KwResult{Subset $R^*\subset N$}
\textbf{Input:} A graph $G=(V,E)$, and $t$\;
\textbf{Output:} Subset $E^*_s\subseteq E$\;
\Begin{
$L_0\longleftarrow diam(G)$\;
$E^*_s\longleftarrow$ optimal solution of Problem \ref{P3_opt}  and $L=L_0$\;
$G_s\longleftarrow (V,E^*_s)$\;
$L_1\longleftarrow diam(G_s)$\;
\While{$L_1>L_0$}{
$L_0\longleftarrow L_1$\;
$E^*_s\longleftarrow$ optimal solution of Problem \ref{P3_opt}  and $L=L_0$\;
$G_s\longleftarrow (V,E^*_s)$\;
$L_1\longleftarrow diam(G_s)$\;
}

\Return{$E^*_s$}
}
\caption{Exact MIP-based Algorithm \label{algorithm_exact}}
\end{algorithm}

\subsection{Edge-weighted graphs}
Here, we show how to generalize Problem \ref{P3_opt} for edge-weighted
graphs with integer weights. Formally, let $w_{ij} \in \mathbb{Z}_+$ denote a positive integer weight (cost) of edge $(i, j) \in E$ 
(for simplicity of exposition, we assume $w_{ij}=0$ for $(i,j)\notin E$ in the corresponding MIP formulation). 
Note that the integrality weights assumption should not be too restrictive in many real-world applications.
Following a similar notation as in Problem \ref{P3_opt}, define $u^{(\ell)}_{ij}=1$ if
and only if there exists a path of length at most $\ell$ in $G_s$,
where $\ell= \{1, . . . , L\}$ and $L \leq (n-1)\max\limits_{(i,j)\in E} w_{ij}$. Then Problem \ref{P3_opt} is
generalized for weighted graphs as follows:

\begin{problem}[edge-weighted path-based MECS]
\label{P3_opt_w}
\begin{subequations}
\begin{align}
&\mbox{minimize} \sum_{(i,j)\in E} x_{ij} \label{f1w_obj}\\
&\mbox{subject to} \nonumber\\
& \sum\limits_{\ell=1}^{L+1} \sum\limits_{i,j=1:i<j}^n \ell(u^{(\ell)}_{ij}-u^{(\ell-1)}_{ij})\leq t\mu_G \frac{n(n-1)}{2} \label{f1w_main}\\
& u^{(1)}_{ij}= x_{ij}, &\hspace{-5mm} i,j\in V: w_{ij}=1\label{f1w-edge_1}\\
& u^{(1)}_{ij}= 0, &\hspace{-5mm}  i,j\in V: w_{ij}\neq 1\label{f1w-edge_2}\\
& u^{(\ell)}_{ij}\geq u^{(\ell-1)}_{ij}, &\hspace{-5mm}\forall i,j\in V,\ \ell\in\{2,\dots,L\}\label{f1w-w_1}\\
& u^{(\ell)}_{ij}\leq x_{ij}+\sum\limits_{k:(i,k)\in E}y^{(\ell)}_{ikj}, &\hspace{-30mm}\forall i,j\in V: w_{ij}\in\{1,\dots,\ell\},\ \ell\in\{2,\dots,L\}\label{f1w-w_21}\\
& u^{(\ell)}_{ij}\geq \sum\limits_{k:(i,k)\in E}y^{(\ell)}_{ikj}, &\hspace{-30mm}\forall i,j\in V: w_{ij}\notin\{1,\dots,\ell\} ,\ \ell\in\{2,\dots,L\}\label{f1w-w_22}\\
& y^{(\ell)}_{ikj}\leq x_{ik},  &\hspace{-40mm}i,j,k\in V: w_{ik}\in \{1,\dots,\ell-1\}, \ell\in\{2,\dots,L\}\label{f1w-y_1}\\
& y^{(\ell)}_{ikj}\leq u^{(\ell-w_{ik})}_{kj},  &\hspace{-40mm}i,j,k\in V: w_{ik}\in \{1,\dots,\ell-1\}, \ell\in\{2,\dots,L\}\label{f1w-y_2}\\
& y^{(\ell)}_{ikj}=0,  &\hspace{-40mm}i,j,k\in V: w_{ik}\geq \ell, \ell\in\{2,\dots,L\}\label{f1w-y_3}\\
&x_{ij},\ u^{(\ell)}_{ij},\ y^{(\ell)}_{ikj} \in \{0,1\}, &\hspace{-5mm}\forall i,j,k,\in V,\ \ell\in\{1,\dots,L\}.\label{f1w-end}
\end{align}
\end{subequations}
\end{problem}

As a final remark, we note that the exact MIP-based Algorithm \ref{algorithm_exact} can be easily generalized for the edge-weighted graphs using 
Problem \ref{P3_opt_w}.

\section{Greedy Strategies}
%!TEX root=main.tex

% spell check has been done on this document
% complete read done

As the MECS problem is $NP$-complete, exact solutions using MIP do not scale very well and so 
exact solutions are hard to find on very large graphs. As a result it is important to have algorithms
for solving the MECS problem on large graphs, if possible exactly and if that is not possible then
at least approximately. In this section, we consider greedy heuristics for the 
MECS problem. We consider several greedy heuristics and we also prove some properties of the resulting
solutions.

\subsection{Greedy Algorithms} 

We start with the simplest greedy algorithm that has several nice properties. The intuition for this algorithm comes
from the observation that we can upper bound the APL, if we can find upper bound for the all pairs shortest paths. 
This is in some sense an overkill. The APL can be bounded above, without explicitly bounding the 
all pairs shortest paths. However, this is a good starting point for our discussions and hence we start with it. 
This algorithm was first proposed for computing standard spanners for both weighted and unweighted graphs by 
Althofar et al.~[\cite{althofer1993sparse}]. First we describe the details of the algorithm and then we describe a few 
properties of the algorithm. As this algorithm is one of the most practical and easy to implement algorithms for
path based spanners, we use this algorithm as a baseline for comparison with our algorithms.

\begin{algorithm}
\DontPrintSemicolon % Some LaTeX compilers require you to use \dontprintsemicolon instead
\SetKwFunction{algo}{GreedySpanner}
\SetKwProg{myalg}{Procedure}{}{}
\myalg{\algo{$G$, $t$}}{
\KwIn{$G=(V,E,W)$ the input graph, $t$ the stretch}
\KwOut{$G'$: The MECS Spanner}
%GreedySpanner($G=(V,E,W)$, $t$){
$G' \gets (V,\{ \})$\;
Sort $E$ in non-decreasing order of weights\;
\For{$e=(u,v) \in E$} {
  $P(u,v) \gets$ Shortest path between $u$ and $v$ in $G'$\;
  \bf{if} {$P(u,v) > t \cdot W(e)$} {\bf then} Add $e$ to $G'$\;
  %\If{$P(u,v) > t \cdot W(e)$} {
  %  Add $e$ to $G'$\;
  %}
}
\Return{$G'$}\;
}
\caption{The Greedy Spanner Algorithm}
\label{algo:greedyspanner}
\end{algorithm}

Going forward we state and prove a few properties of the algorithm. We first prove that the algorithm always gives 
a feasible solution to the MECS problem. Then we state a few results that follow directly from the properties of the 
algorithm. Interested readers may refer to [\cite{althofer1993sparse}] for detailed proofs of the same.

\begin{lemma} \label{l1} 
Algorithm GreedySpanner always gives a feasible solution to the MECS problem.
\end{lemma}

\begin{proof}
First consider the edges in $G-G'$. Algorithm [\ref{algo:greedyspanner}] chose to ignore these edges because for each
such edge $e=(u,v)$, $P(u,v) \leq stretch \cdot W(e)$. Consider any two vertices $a,b \in V$ and consider the shortest
path $P_G(a,b)$ between them in $G$. Consider the edges $e=(u,v)$ on this path such that $e \in G-G'$. Each such edge is
replaced in $G'$ by a path $P_{G'}(u,v)$ such that $P_{G'}(u,v) \leq stretch \cdot W(e)$. Moreover for edges $e=(u,v)$
on $P_G(a,b)$ such that $e \in G'$, $W(e) \leq stretch \cdot W(e)$ as $stretch \geq 1$. Thus we have that $P_G(a,b) \leq
stretch \cdot P_{G'}(a,b)$. This is true for all pairs of vertices $a,b \in G$. Thus $\sum_{a,b \in G} P_G(a,b) \leq
stretch \cdot \sum_{a,b \in G'} P_{G'}(a,b)$. As the vertex sets of $G$ and $G'$ are the same, this completes the
proof.
\end{proof}

The following lemmas make statements about the structure of the solution returned by the greedy spanner algorithm.
The first result relates to the structure of the resulting subgraph and the next result gives bounds on the size and weight
of the resulting solution. Interested readers may consult~[\cite{althofer1993sparse}] for detailed proofs of these results.
The following lemma states that the solution returned by the algorithm [\ref{algo:greedyspanner}] contains the minimum 
spanning tree as a subgraph.

\begin{lemma} 
The graph $G'$ contains the MST of $G$ as a subgraph. 
\end{lemma}

Next we state a theorem that states that the size and weight of the solution obtained from algorithm [\ref{algo:greedyspanner}] 
are bounded above and gives the values of the upper bounds. One of the nice properties of the aforementioned algorithm is the 
fact that the weight of the resulting spanner is bounded above by a constant multiple of the weight of the MST.

\begin{theorem} 
Given a weighted graph $G=(V,E,W)$ and a stretch factor $t > 0$ it is possible to construct a feasible
solution to the MECS problem for stretch $2t+1$ such that the solution $G'$ has the following
properties: 
\begin{enumerate} 
\item $Size(G') < |V| \lceil |V|^{\frac{1}{t}} \rceil$ where \emph{size} counts the number of edges in the graph 
\item $Weight(G') < Weight(MST(G))(1+\frac{|V|}{2t})$ 
\end{enumerate} 
\end{theorem}

Algorithm [\ref{algo:greedyspanner}] maintains a forest of connected components at any time and at every step adds 
an edge to merge two connected components. At the end of the algorithm we are left with a single connected component 
that satisfies the spanner condition. The algorithm works on \emph{local} information. At each step, it looks at the 
shortest path length between every pair of vertices. It makes sure that this is bounded above by a constant multiple of
the shortest path length between the vertices in the input graph. This in turn preserves the APL. However, this is more
than what we need. This algorithm preserves the shortest path length between all the vertex pairs and hence also 
preserves the original diameter to a constant multiple. We just need the preservation of the average path length and thus
this algorithm does more than what we require.

Going forward, we state two greedy heuristics for the MECS problem. Both of them do not use any local information.
The first algorithm starts with the original graph and removes edges from it, one by one, making sure that every edge
removal preserves the constraint on the APL. The edges are considered in the decreasing order of weights, thus making this
strategy a greedy strategy. The second algorithm mimics the algorithm of Althofar et al.~[\ref{algo:greedyspanner}] and 
starts with a forest. At each iteration, an edge is added, merging two components. The edges are considered in increasing 
order of weights and this process continues as long as the condition on the APL is violated. We state and analyze these algorithms 
next.

\begin{algorithm}
\DontPrintSemicolon % Some LaTeX compilers require you to use \dontprintsemicolon instead
\SetKwFunction{algo}{GreedyMECS}
\SetKwProg{myalg}{Procedure}{}{}
\myalg{\algo{$G$, $t$}}{
\KwIn{$G=(V,E,W)$ the input graph, $t$ the stretch}
\KwOut{$G$: The MECS Spanner}
%GreedySpanner($G=(V,E,W)$, $t$){
$\mu \gets$ average distance in graph $G$\;
Sort $E$ in non-increasing order of weights\;
\For{$e \in E$} {
  $G_{Temp} \gets G-{e}$\;
  $\mu_{Temp} \gets$ average distance of $G_{Temp}$\;
  {\bf if} {$\mu_{temp} > t \cdot \mu$} {\bf then} Do not remove $e$ from $G$ {\bf else} $G \gets G-{e}$\;
  %\uIf{$\mu_{temp} > t \cdot \mu$} {
  %  Do not remove $e$ from $G$\;
  %}
  %\Else{
  %  $G \gets G-{e}$\;
  %}
}
\Return{$G$}\;
}
\caption{The Greedy Removal MECS Algorithm}
\label{algo:greedymecs}
\end{algorithm}

\begin{algorithm}
\DontPrintSemicolon % Some LaTeX compilers require you to use \dontprintsemicolon instead
\SetKwFunction{algo}{GreedyMECSV2}
\SetKwProg{myalg}{Procedure}{}{}
\myalg{\algo{$G$, $t$}}{
\KwIn{$G=(V,E,W)$ the input graph, $t$ the stretch}
\KwOut{$G$: The MECS Spanner}
%GreedySpanner($G=(V,E,W)$, $t$){
$G' \gets (V,\{ \})$\;
Sort $E$ in non-decreasing order of weights\;
$\mu \gets$ average distance in graph $G$\;
\For{$e \in E$} {
  $\mu_{G'} \gets$ average distance of $G'$\;
  %$\mu_{Temp} \gets$ average distance of $G_{Temp}$\;
  {\bf if} {$\mu_{G'} > t \cdot \mu$} {\bf then} $G' \gets G' \cup {e}$\;
  %\If{$\mu_{G'} > t \cdot \mu$} {
  %  $G' \gets G' \cup {e}$\;
  %}
}
\Return{$G'$}\;
}
\caption{The Greedy Addition MECS Algorithm}
\label{algo:greedymecsv2}
\end{algorithm}

One implicit assumption of algorithm [\ref{algo:greedymecs}] is that the removal of the edge $e$ from $G$ does not
leave the graph disconnected. In implementations, if the removal of the edge $e$ leaves the graph disconnected, then the
APL is distorted infinitely and hence the edge is not removed. Next we state a theorem regarding the nature of
the solution returned by the greedy algorithms. The proof of the theorem is provided as separate lemmas in Appendix B.

\begin{theorem} 
The following statements hold for the solutions returned by the algorithms [\ref{algo:greedymecs}] and [\ref{algo:greedymecsv2}]:
\begin{enumerate}
\item Algorithm [\ref{algo:greedymecs}]  always returns a feasible solution to the MECS problem
\item Consider the graph $G'$ returned by the algorithm [\ref{algo:greedymecs}] . For any input graph $G$, the graph $G'$
contains the MST of $G$
\item Algorithm [\ref{algo:greedymecsv2}]  always gives a feasible solution to the MECS problem
\item Consider the graph $G'$ returned by the algorithm [\ref{algo:greedymecsv2}]. For any input graph $G$, the graph $G'$
contains the MST of $G$
\end{enumerate}
\end{theorem}

%\begin{lemma} 
%Algorithm [\ref{algo:greedymecs}]  always returns a feasible solution to the MECS problem.
%\end{lemma}
%The same statement can be made about the solution returned by the algorithm [\ref{algo:greedymecsv2}]. We state this
%as a separate lemma.
%\begin{lemma} 
%Algorithm [\ref{algo:greedymecsv2}]  always gives a feasible solution to the MECS problem.
%\end{lemma}
%Next we consider two lemmas that describe the structure of the solution returned by the algorithms [\ref{algo:greedymecs}]
%and [\ref{algo:greedymecs}]. More precisely, we claim that the MECS solutions returned by these algorithms will always
%contain the MST of the underlying graph as a subgraph.
%\begin{lemma}
%Consider the graph $G'$ returned by the algorithm [\ref{algo:greedymecs}] . For any input graph $G$, the graph $G'$
%contains the MST of $G$. 
%\end{lemma}
%The next lemma makes the same assertion about the solution returned by the algorithm [\ref{algo:greedymecsv2}].
%\begin{lemma}
%Consider the graph $G'$ returned by the algorithm [\ref{algo:greedymecsv2}] . For any input graph $G$, the graph $G'$
%contains the MST of $G$. 
%\end{lemma}

Apart from being an interesting observation, with a nice and concise proof, the above result also helps us to change the
algorithm [\ref{algo:greedymecsv2}]. As we know that the solution returned by the algorithm [\ref{algo:greedymecsv2}] will 
always contain the MST, instead of starting the algorithm from a forest, we can start the algorithm from the MST. Thus we 
compute the MST of the input graph and then we start adding edges that are not included in the MST, in increasing order of 
weights. We continue this process until the APL of the resulting graph satisfies the constraint. These observations are also
discussed in Appendix B due to space constraints.

\section{Experimental Results}
%!TEX root=main.tex

% spell check has been done on this document

% complete read done

In an attempt to assess the performance of the proposed algorithms in comparison with others in the literature, we performed 
computational experiments using several standard data sets. In this section, we present and discuss our findings from these
experiments.

In order to compare the different algorithms and formulations, we implemented and tested our algorithms on several different 
real life networks. The flow based approach, the path based approach and the MIP based approach were all implemented using 
FICO Xpress optimization suite~[\cite{bussieck2010minlp}]. The greedy heuristics were implemented using Python~[\cite{van2007python}], 
using the Networkx~[\cite{schult2008exploring}] library to implement the standard graph algorithms. All the implementations were tested
against the Greedy Spanner algorithm~[\ref{algo:greedyspanner}] of Althofar et al.~[\cite{althofer1993sparse}], which, as has been shown 
above, computes a feasible solution to the MECS problem. However, this solution is not always the optimal one.

All the experiments were carried out on a laptop running OSX Maverics, having 8 GB RAM and an Intel Core i7 processor. In all our implementations 
we use the initial average path length~($\mu_G$) and the target average path length~($\mu_{G_s}$) as the input along with the input graph. We do not 
use the stretch factor explicitly and it is computed based on the two inputs. Going forward we first describe the datasets that we used and finally 
we describe the results of our experiments.

\subsection{Datasets}
Most of our experiments were done with unweighted graphs, though the underlying algorithms might as well work with weighted graphs. The main reason
for not using arbitrary weighted graphs for the experiments was the observation that solutions based on direct optimization methods did not scale very
well with weighted graphs. We do however, report the results of our experiments on weighted unit disk graphs~[\cite{clark1990unit}]. The first dataset 
that we use is the famous karate club graph~[\cite{zachary1977information}]. This is a social network graph that was first collected in 1977 and depicts
the friendship relations between 34 members of a karate club. 
%The second graph that we used is that of the IEEE 118-bus powergrid network~[\cite{ieee118bus}]. 
We also used our algorithms on the Kreb's network graph. As mentioned before we ran our experiments on Unit disk weighted graphs. These
graphs were generated randomly and we used both a weighted as well as an unweighted version of the graph. As our goal is to minimize the total number 
of edges~(or total edge weight) in the MECS spanner, we report this number for the different algorithms. We compare the results from the greedy spanner 
algorithm against the optimal solutions obtained by our integer programming based approaches as well as against the results from the greedy heuristics.

\subsection{Results}
The results obtained by our algorithms on the different datasets are described below. As mentioned at the beginning of this section, instead of 
explicitly using a \emph{stretch}, we use an increment. Thus if the original APL for the input graph is $\mu$ and the target increment is $\delta$,
then the goal is to find a sparse subgraph of the input graph whose APL is at most $\mu + \delta$. In all our experiments we use three values of the 
increments namely, 0.1, 0.2 and 0.3. Next we describe the results for each of the datasets separately.

\subsubsection{Karate Club Graph}
The results of running the MIP based algorithms on the Karate Club graph are shown in figure [\ref{fig_karate_mip}]. The original graph has an average
path length of $\mu = 2.41$. The figures show the results of running the algorithm with target average path lengths of 2.51, 2.59 and 2.69. We note that
when the target APL is 2.51, the optimal number of edges is 48 and the number of edges decreases to 37 when the target is 2.69. The total number of 
edges in the minimum spanning tree for the graph is 33. The results for this graph are summarized in the table \ref{karate}.

%For the same graph using the greedy algorithm of Althofar et al.~[\cite{althofer1993sparse}], 
%if the APL increment is 0.1, then the size of the resulting spanner is 78, which imples that the algorithm essentially returns the same graph. 
%Using the same target, when using the algorithm~[\ref{algo:greedymecs}] which \emph{removes} edges one by one, the resulting graph has a size of 73. 
%For the variation where the edges are added one by one to an empty graph, as in algorithm [\ref{algo:greedymecsv2}], the size of the resulting 
%spanner is 72. When the target APL has an increment of 0.2, the algorithm of Althofar et al.~[\cite{althofer1993sparse}] still returns the 
%original graph. However, the algorithm [\ref{algo:greedymecs}] returns a graph of size 71 and the algorithm [\ref{algo:greedymecsv2}] 
%returns a graph of size 67. Finally with an increment of 0.3, the size of the resulting spanner from the algorithm [\ref{algo:greedymecs}] 
%is 66 whereas from algorithm [\ref{algo:greedymecsv2}] it is 67. However, the output from the standard spanner algorithm is still the original 
%graph. 

Though the results from the algorithms [\ref{algo:greedymecs}] and [\ref{algo:greedymecsv2}], are not optimal, they still perform 
better than the greedy spanner algorithm of Althofar et al. We also note that empirically, the size of the spanners returned by the greedy 
algorithms is at most twice the size of the optimal spanner. We also note that the size of the MST is 33. An interesting thing happens when 
the increment to the average distance is set to be at 2. Then the algorithm [\ref{algo:greedymecs}] results in a spanner that has the same 
size as the MST, that is 33. However, the algorithm [\ref{algo:greedymecsv2}] still gives a spanner of size 67 and the algorithm of 
Althofar et al. still returns the original graph. The figure [\ref{fig_karate_greedy}] shows the results of running the different variations of 
the greedy algorithms on the Karate club graph for an increment of 0.3. 

\begin{table}[ht]
\centering
 \begin{tabular}{||c c c c c||} 
 \hline
 Increment & MIP & Greedy1 & Greedy2 & Althofar \\ [0.5ex] 
 \hline\hline
 0.1 & 48 & 73 & 72 & 78 \\ 
 \hline
 0.2 & 40 & 71 & 67 & 78 \\
 \hline
 0.3 & 37 & 66 & 67 & 78 \\
 \hline
\end{tabular}
\caption{Size of spanners for the karate club graph}
\label{karate}
\end{table}

\begin{figure}[H]
  \centering
\subfloat[$|E|=78$, $\mu=2.41$]{\label{fig_karate_orig}\includegraphics[scale=0.2]{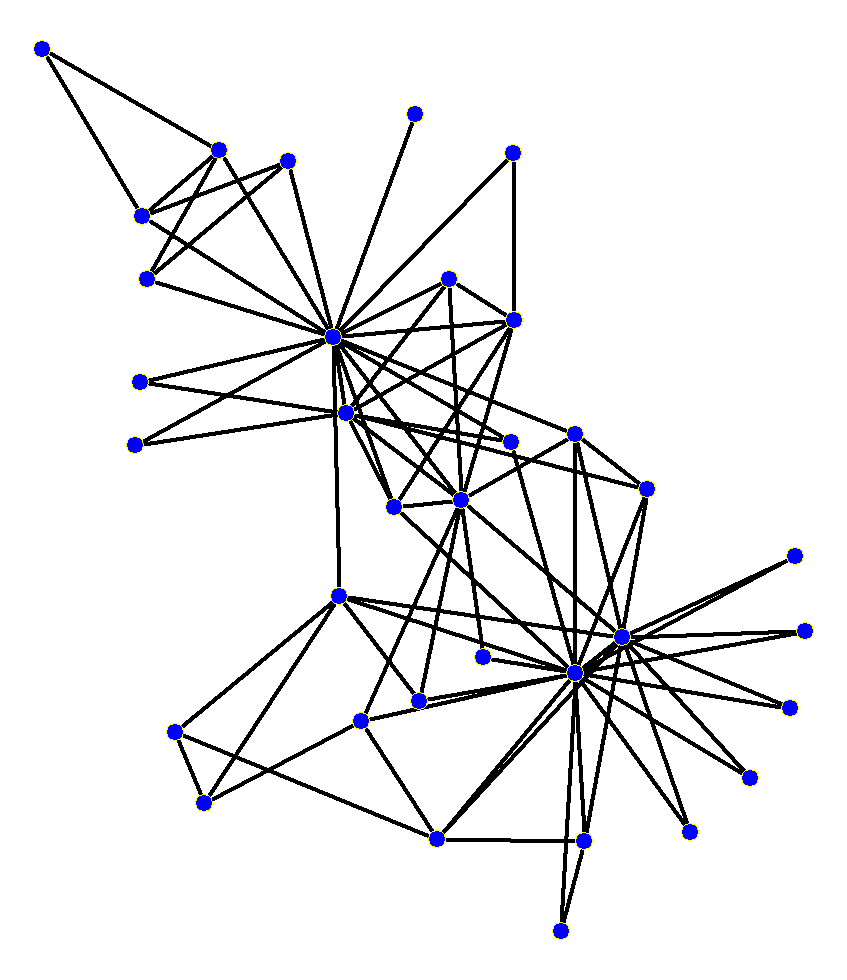}}
\subfloat[ $|E_s|=48$, $\mu_s=2.51$]{\label{fig_karate_25}\includegraphics[scale=0.2]{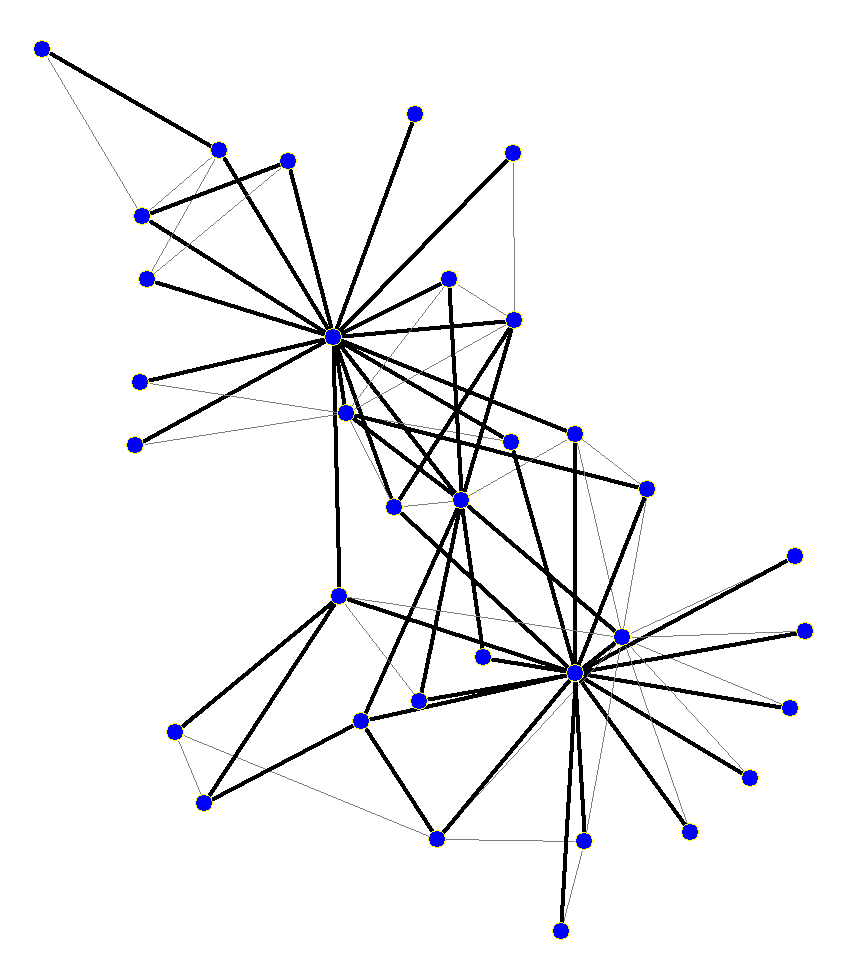}} \\
\subfloat[ $|E_s|=40$, $\mu_s=2.61$]{\label{fig_karate_26}\includegraphics[scale=0.2]{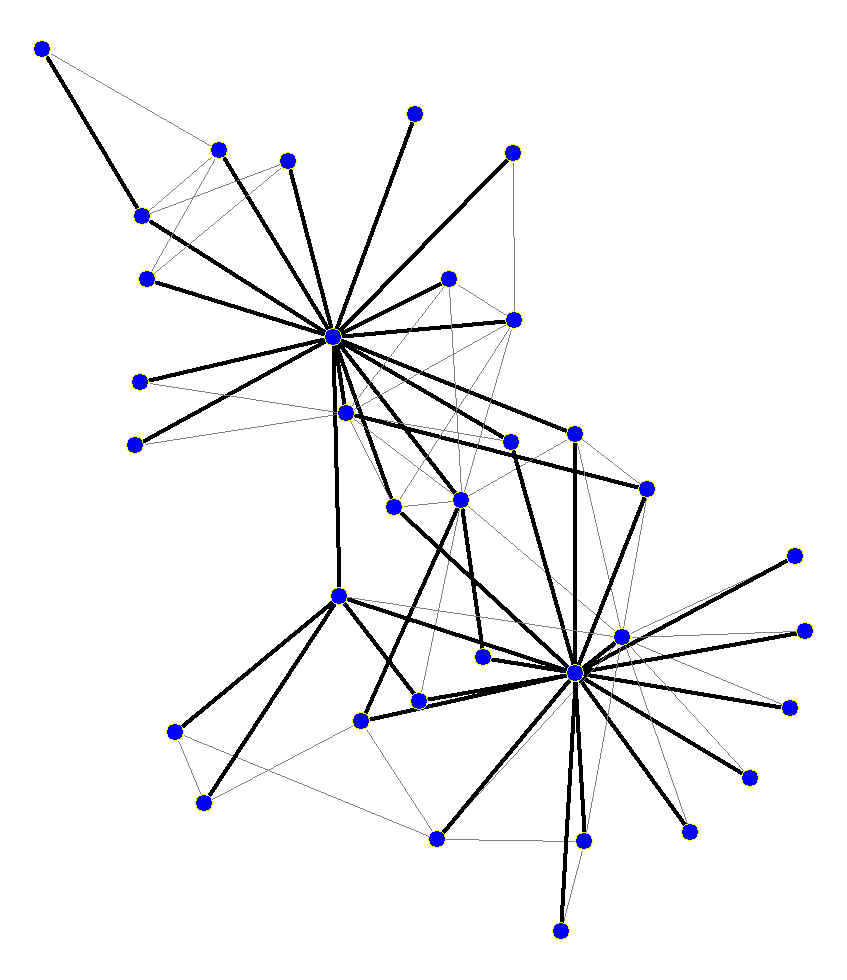}}
\subfloat[ $|E_s|=37$, $\mu_s=2.71$]{\label{fig_karate_27}\includegraphics[scale=0.2]{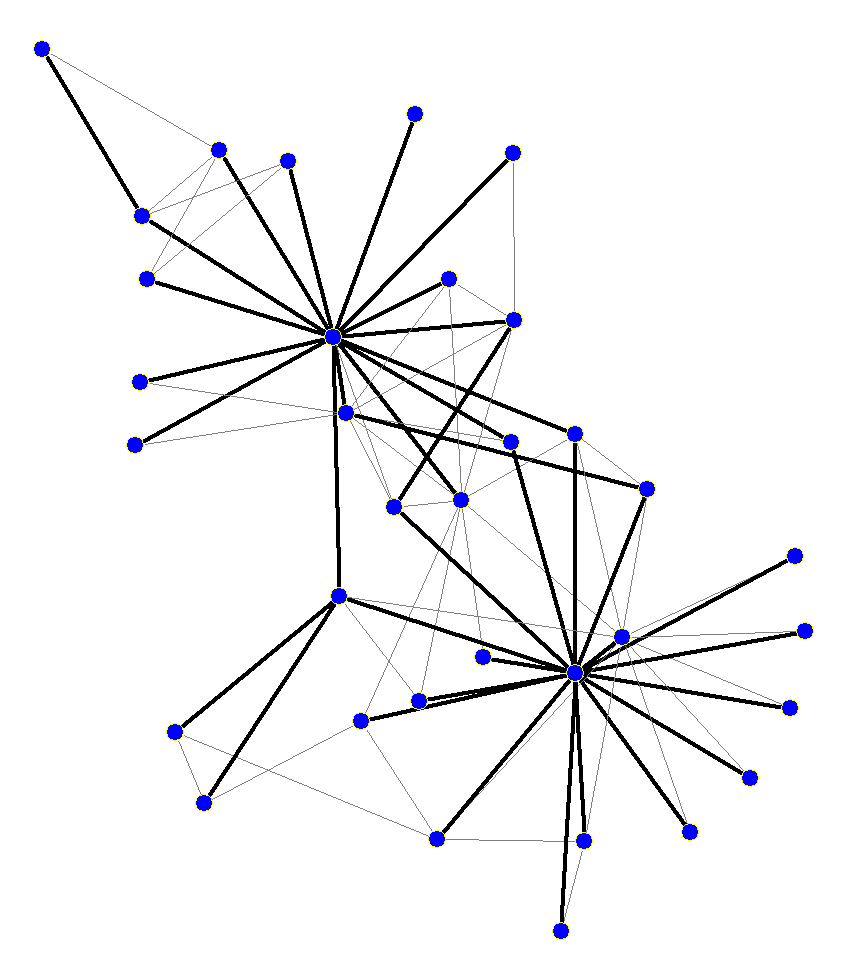}}
  \caption{Illustration of the (a) \textbf{Karate Club} network ($|V|=34, |E|=78$) with average distance $\mu=2.41$ 
  (b) the optimal compact spanner (thick edges) with average distance $\mu_s=\mu+0.1$ (c) the optimal compact spanner 
  (thick edges) with average distance $\mu_s=\mu+0.2$ (d) the optimal compact spanner (thick edges) with average distance $\mu_s=\mu+0.3$}
    \label{fig_karate_mip}
\end{figure}

\begin{figure}[H]
  \centering
\subfloat[$|E|=66$, $\mu_s=2.71$]{\label{fig_karate_greedy_03_r}\includegraphics[width=0.5\textwidth]{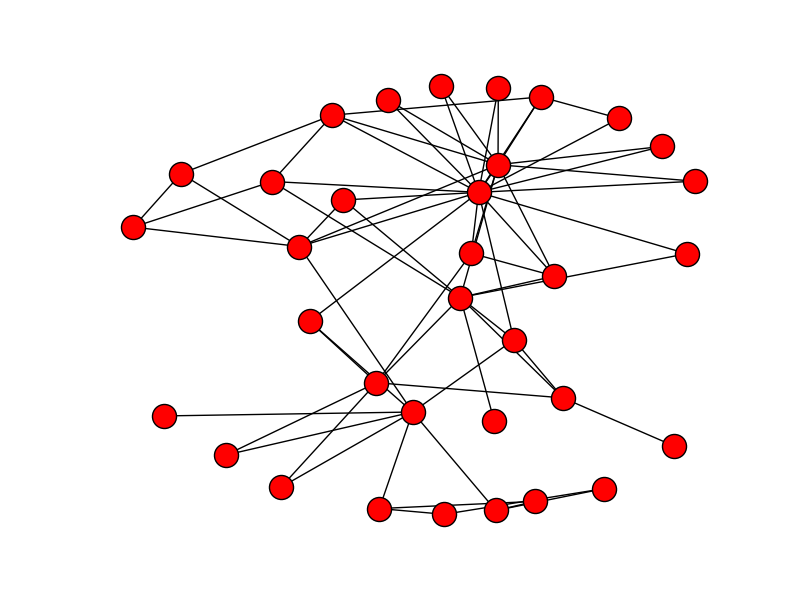}}
\subfloat[ $|E_s|=67$, $\mu_s=2.71$]{\label{fig_karate_greedy_03_a}\includegraphics[width=0.5\textwidth]{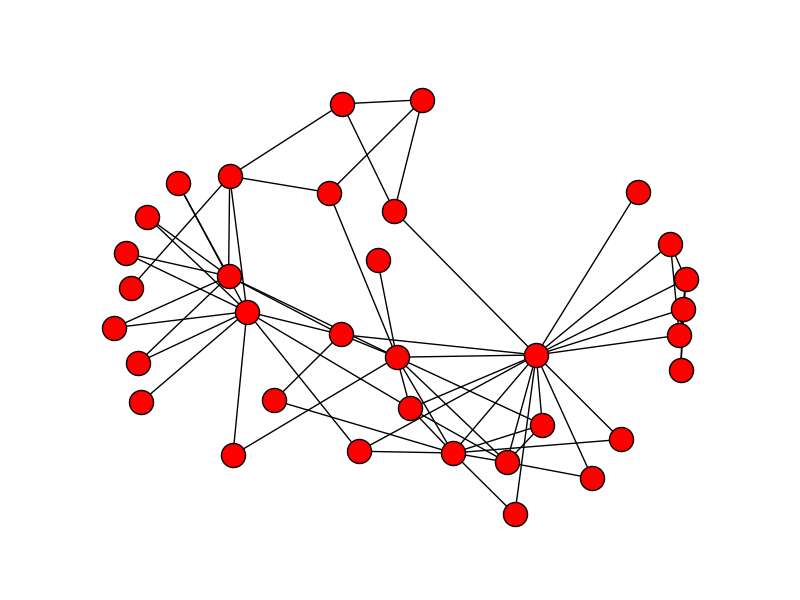}} \\
  \caption{Illustration of the (a) \textbf{Karate Club} network: result of greedy algorithm [\ref{algo:greedymecs}] with increment of 0.3 
  (b) \textbf{Karate Club} network: result of greedy algorithm [\ref{algo:greedymecsv2}] with increment of 0.3 }
    \label{fig_karate_greedy}
\end{figure}

%\subsubsection{IEEE Bus-118 Network}

%\begin{figure}[t]
%  \centering
%\subfloat[$|E|=179$, $\mu=6.31$]{\label{fig_ieee_orig}\includegraphics[scale=0.3]{images/ieee_orig.png}}
%\subfloat[ $|E_s|=136$, $\mu_s=6.41$]{\label{fig_ieee_641}\includegraphics[scale=0.3]{images/ieee_641.png}}
%  \caption{Illustration of the (a) \textbf{IEEE bus 118} network ($|V|=118, |E|=179$) with average distance $\mu=6.31$ and 
% (b) the optimal compact spanner (thick edges) with average distance $\mu_s=\mu+0.1$.}
%    \label{fig_ieee_mip}
%\end{figure}

\subsubsection{Unit Disk Graph}
Next we present the results for running our algorithms on an artificially generated graph. We generated the Unit Disk Graph following a standard 
procedure in wireless communication network analysis: 
\begin{itemize}
\item First we generate a 100x100 box and put 50 \enquote{sensors} inside the box
\item Set a communication range of 20, thus if the sensors are less than 20 units apart then they can communicate directly
\end{itemize}
As before we consider the results of running our algorithms for the increments of 0.1, 0.2 and 0.3. The results obtained are shown in the table
[\ref{unwt_ud}]. Its obvious that for small increments, as we have considered, the algorithm of Althofar et al. is not able to sparsify the graph
at all. However, the algorithm [\ref{algo:greedymecs}] performs well and for each of the increments it gives us some sparsification of the
underlying graph. However, the algorithm based on the addition of edges, [\ref{algo:greedymecsv2}], does not perform as well as the first for 
small increments.

\begin{table}[ht]
\centering
 \begin{tabular}{||c c c c c||} 
 \hline
 Increment & MIP & Greedy1 & Greedy2 & Althofar \\ [0.5ex] 
 \hline\hline
 0.1 & 68 & 103 & 107 & 119 \\ 
 \hline
 0.2 & 62 & 99 & 106 & 119 \\
 \hline
 0.3 & 58 & 89 & 106 & 119 \\
 \hline
\end{tabular}
\caption{Size of spanners for the unweighted unit disk graph}
\label{unwt_ud}
\end{table}

The table [\ref{wt_ud_1}] shows the results of running our algorithms on a weighted unit disk graph. The weights on the edges are either 1 or 2 
depending on the Euclidean distance between the vertices. The reason for constructing this graph in this particular way was to be able to run 
the MIP based algorithms on this graph as a matter of comparison with the greedy algorithms. We were not able to run the MIP based solutions
on general weighted graphs of this size. Our main goal in this experiment was to be able to compare the quality of the solutions obtained for 
weighted graphs, from both the MIP based algorithms and the greedy algorithms. %In the table [\ref{wt_ud}] we only show the size of the spanners 
%obtained using the different algorithms. We again note that the gredy algorithm based on removing edges from the input graph performs better 
%than the other greedy algorithms. 
Again it is clear that algorithm [\ref{algo:greedymecs}] outperforms the other two greedy algorithms.

%\begin{table}[h!]
%\centering
% \begin{tabular}{||c c c c c||} 
% \hline
% Increment & MIP & Greedy1 & Greedy2 & Althofar \\ [0.5ex] 
% \hline\hline
% 0.1 & * & 115 & 171 & 125 \\ 
% \hline
% 0.2 & * & 110 & 171 & 125 \\
% \hline
% 0.3 & * & 107 & 161 & 125 \\
% \hline
%\end{tabular}
%\caption{Size of spanners for the weighted unit disk graph}
%\label{wt_ud}
%\end{table}

\begin{table}[ht]
\centering
 \begin{tabular}{||c c c c c||} 
 \hline
 Increment & MIP & Greedy1 & Greedy2 & Althofar \\ [0.5ex] 
 \hline\hline
 0.1 & 119 & 169 & 281 & 189 \\ 
 \hline
 0.2 & 102 & 159 & 281 & 189 \\
 \hline
 0.3 & 93 & 153 & 261 & 189 \\
 \hline
\end{tabular}
\caption{Weight of spanners for the weighted unit disk graph}
\label{wt_ud_1}
\end{table}

The resulting spanner graphs obtained for the unweighted and the weighted versions of the Unit Disk Graph are shown in
the Appendix C.
%[\ref{AppC}].

\subsubsection{Krebs Network}
The results of running our algorithms on the krebs36 network is given in table [\ref{krebs}]. The size of the original input graph is 153 and the APL
is 2.92. The size of the MST for this graph is 61. As before we consider three values for the increments, namely, 0.1, 0.2 and 0.3. Among the greedy
algorithms, we observe that the algorithm [\ref{algo:greedymecs}] outperforms the other two greedy algorithms.

\begin{table}[ht]
\centering
 \begin{tabular}{||c c c c c||} 
 \hline
 Increment & MIP & Greedy1 & Greedy2 & Althofar \\ [0.5ex] 
 \hline\hline
 0.1 & 82 & 139 & 149 & 153 \\ 
 \hline
 0.2 & 63 & 132 & 149 & 153 \\
 \hline
 0.3 & 61 & 136 & 149 & 153 \\
 \hline
\end{tabular}
\caption{Size of spanners for the unweighted krebs graph}
\label{krebs}
\end{table}

\subsubsection{Discussion}
We conclude this section with a small explanation of the results that we observe with the greedy algorithms. 
\paragraph{Greedy MECS Algorithms} 
As seen from the results of the experiments above, the algorithm [\ref{algo:greedymecs}] outperforms the algorithm [\ref{algo:greedymecsv2}] in most of the 
cases. In order to understand why this happens one needs to look at the way the two algorithms operate. The first one \emph{removes} edges from the graph
whereas the second one \emph{adds} edges to the graph. However, when removing the edges [\ref{algo:greedymecs}] only removes edges that 
are \emph{relevant}. An edge is relevant for removal if the resulting subgraph has an APL that satisfies the MECS criteria. Thus at every iteration
[\ref{algo:greedymecs}] removes only the most important edges. On the other hand [\ref{algo:greedymecsv2}] starts with a graph whose average path length
is infinite~(namely a set of disconnected points). At each step it adds and edge if the current subgraph violates the MECS criteria. Thus it will go on
doing this until the MECS criteria is satisfied. At this point it will stop the addition of edges. Thus it may end up adding more edges than is required
to satisfy the given APL target. The only way around this problem is to add only relevant edges, edges that bring down the APL the most among all the 
edges that have not been considered till now. One of the advantages of this algorithm is that most often it achieves a APL value that is lower than the
target, however at the cost of increasing the size of the resulting spanner. However, that would complicate the implementation more. We postpone 
such studies to a later paper.

\paragraph{Greedy Spanner Algorithm} 
We have observed that for unweighted graphs the greedy spanner algorithm does not return a sparse graph for most of our test cases. We have shown before 
that the solution returned by this algorithm is feasible for the MECS problem, however because of this issue, this algorithm does not present us with a 
practical option for computing MECS spanners on unweighted graphs. In order to understand why this problem occurs, we need to realize that there is a 
distinct difference between the magnitude of the stretch factor that is used for path based spanners and the APL based MECS spanners. In our case, 
we want the stretch to be really small, we don't want the average path length to increase too much. As a result for most of our experiments the 
stretch is in the interval $(1,1+\epsilon)$ for small $\epsilon > 0$. For unweighted graphs, this causes a problem, as the path length can only 
change in integral multiples of unity. Thus depending on the value of the stretch that we use, we would get some sparsification against no 
sparsification at all. This in turn justifies effort to look for more efficient algorithms for computing near optimal MECS solutions.

\section{Conclusions and Future Work}
In this paper we have introduced the \emph{Minimum Edge Compact Spanner} problem and we have shown
that the problem is $NP$-hard. We have used several MIP formulations of the problem to get
the optimal MECS solutions. Moreover we have proposed two greedy algorithms for the MECS problem
and have used them on real life networks. We have compared the results obtained from the MIP
based solutions and the greedy algorithms. 

In the near future, we would like to study the following problems:
%this problem in more detail. More precisely, we are interested 
%in answering the following questions: 
(1) is it possible to prove that the algorithm [\ref{algo:greedymecs}]
gives a constant factor approximation for the MECS problem (2) are there other approximation algorithms
for the MECS problem that gives $O(\log |V|)$-factor approximations (3) does the MECS problem have a PTAS
and a FPTAS (4) if not then what is the hardness of approximating the MECS problem. We would like to consider
these problems in the settings of Euclidean graphs, graphs on metric spaces that have nice packing properties
like doubling metrics as well as for general graphs~(metric spaces with shortest path metric). Average
path length is an important parameter for graphs and there are several other similar parameters. We would
like to investigate whether similar problems can be solved for these parameters as well.

\newpage
%\section*{References}

\bibliography{Alex,tm}

\begin{thebibliography}{10}
\expandafter\ifx\csname url\endcsname\relax
  \def\url#1{\texttt{#1}}\fi
\expandafter\ifx\csname urlprefix\endcsname\relax\def\urlprefix{URL }\fi
\expandafter\ifx\csname href\endcsname\relax
  \def\href#1#2{#2} \def\path#1{#1}\fi

\bibitem{ellens2013graph}
W.~Ellens, R.~E. Kooij, Graph measures and network robustness, arXiv preprint
  arXiv:1311.5064.

\bibitem{jeong2000large}
H.~Jeong, B.~Tombor, R.~Albert, Z.~N. Oltvai, A.-L. Barab{\'a}si, The
  large-scale organization of metabolic networks, Nature 407~(6804) (2000)
  651--654.

\bibitem{kirkby1976tests}
M.~Kirkby, Tests of the random network model, and its application to basin
  hydrology, Earth Surface Processes 1~(3) (1976) 197--212.

\bibitem{watts1998collective}
D.~J. Watts, S.~H. Strogatz, Collective dynamics of small-world networks,
  nature 393~(6684) (1998) 440--442.

\bibitem{newman2002random}
M.~E. Newman, D.~J. Watts, S.~H. Strogatz, Random graph models of social
  networks, Proceedings of the National Academy of Sciences 99~(suppl 1) (2002)
  2566--2572.

\bibitem{adamic2000power}
L.~A. Adamic, B.~A. Huberman, Power-law distribution of the world wide web,
  science 287~(5461) (2000) 2115--2115.

\bibitem{bondy1976graph}
J.~A. Bondy, U.~S.~R. Murty, Graph theory with applications, Vol. 290,
  Macmillan London, 1976.

\bibitem{chung1988average}
F.~Chung, The average distance and the independence number, Journal of Graph
  Theory 12~(2) (1988) 229--235.

\bibitem{althofer1993sparse}
I.~Althofer, G.~Das, D.~Dobkin, D.~Joseph, J.~Soares, On sparse spanners of
  weighted graphs, Discrete \& Computational Geometry 9~(1) (1993) 81--100.

\bibitem{peleg1989graph}
D.~Peleg, A.~A. Sch{\"a}ffer, Graph spanners, Journal of graph theory 13~(1)
  (1989) 99--116.

\bibitem{peleg1987optimal}
D.~Peleg, J.~D. Ullman, An optimal synchronizer for the hypercube, in:
  Proceedings of the sixth annual ACM Symposium on Principles of distributed
  computing, ACM, 1987, pp. 77--85.

\bibitem{cai1995tree}
L.~Cai, D.~G. Corneil, Tree spanners, SIAM Journal on Discrete Mathematics
  8~(3) (1995) 359--387.

\bibitem{emek2004approximating}
Y.~Emek, D.~Peleg, Approximating minimum max-stretch spanning trees on
  unweighted graphs, in: Proceedings of the fifteenth annual ACM-SIAM symposium
  on Discrete algorithms, Society for Industrial and Applied Mathematics, 2004,
  pp. 261--270.

\bibitem{elkin2000strong}
M.~Elkin, D.~Peleg, Strong inapproximability of the basic k-spanner problem,
  in: Automata, Languages and Programming, Springer, 2000, pp. 636--648.

\bibitem{elkin2000hardness}
M.~Elkin, D.~Peleg, The hardness of approximating spanner problems, in: STACS
  2000, Springer, 2000, pp. 370--381.

\bibitem{elkin2005approximating}
M.~Elkin, D.~Peleg, Approximating k-spanner problems for k> 2, Theoretical
  Computer Science 337~(1) (2005) 249--277.

\bibitem{levcopoulos2002improved}
C.~Levcopoulos, G.~Narasimhan, M.~Smid, Improved algorithms for constructing
  fault-tolerant spanners, Algorithmica 32~(1) (2002) 144--156.

\bibitem{chechik2010fault}
S.~Chechik, M.~Langberg, D.~Peleg, L.~Roditty, Fault tolerant spanners for
  general graphs, SIAM Journal on Computing 39~(7) (2010) 3403--3423.

\bibitem{dinitz2011fault}
M.~Dinitz, R.~Krauthgamer, Fault-tolerant spanners: better and simpler, in:
  Proceedings of the 30th annual ACM SIGACT-SIGOPS symposium on Principles of
  distributed computing, ACM, 2011, pp. 169--178.

\bibitem{solomon2012fault}
S.~Solomon, Fault-tolerant spanners for doubling metrics: Better and simpler,
  arXiv preprint arXiv:1207.7040.

\bibitem{solomon2014hierarchical}
S.~Solomon, From hierarchical partitions to hierarchical covers: Optimal
  fault-tolerant spanners for doubling metrics, in: Proceedings of the 46th
  Annual ACM Symposium on Theory of Computing, ACM, 2014, pp. 363--372.

\bibitem{chan2012sparse}
T.-H.~H. Chan, M.~Li, L.~Ning, Sparse fault-tolerant spanners for doubling
  metrics with bounded hop-diameter or degree, in: Automata, Languages, and
  Programming, Springer, 2012, pp. 182--193.

\bibitem{bose2013robust}
P.~Bose, V.~Dujmovic, P.~Morin, M.~Smid, Robust geometric spanners, SIAM
  Journal on Computing 42~(4) (2013) 1720--1736.

\bibitem{cormen2009introduction}
T.~H. Cormen, Introduction to algorithms, MIT press, 2009.

\bibitem{hwang1992steiner}
F.~K. Hwang, D.~S. Richards, P.~Winter, The Steiner tree problem, Vol.~53,
  Elsevier, 1992.

\bibitem{botton2013benders}
Q.~Botton, B.~Fortz, L.~Gouveia, M.~Poss, Benders decomposition for the
  hop-constrained survivable network design problem, INFORMS journal on
  computing 25~(1) (2013) 13--26.

\bibitem{awerbuch1997buy}
B.~Awerbuch, Y.~Azar, Buy-at-bulk network design, in: Foundations of Computer
  Science, 1997. Proceedings., 38th Annual Symposium on, IEEE, 1997, pp.
  542--547.

\bibitem{ma2016minimum}
J.~Ma, F.~M. Pajouh, B.~Balasundaram, V.~Boginski, The minimum spanning k-core
  problem with bounded cvar under probabilistic edge failures, INFORMS Journal
  on Computing 28~(2) (2016) 295--307.

\bibitem{golden2005heuristic}
B.~Golden, S.~Raghavan, D.~Stanojevi{\'c}, Heuristic search for the generalized
  minimum spanning tree problem, INFORMS Journal on Computing 17~(3) (2005)
  290--304.

\bibitem{melkonian2005primal}
V.~Melkonian, {\'E}.~Tardos, Primal-dual-based algorithms for a directed
  network design problem, INFORMS Journal on Computing 17~(2) (2005) 159--174.

\bibitem{gupta2011approximation}
A.~Gupta, J.~K{\"o}nemann, Approximation algorithms for network design: A
  survey, Surveys in Operations Research and Management Science 16~(1) (2011)
  3--20.

\bibitem{johnson1978complexity}
D.~S. Johnson, J.~K. Lenstra, A.~Kan, The complexity of the network design
  problem, Networks 8~(4) (1978) 279--285.

\bibitem{wong1980worst}
R.~T. Wong, Worst-case analysis of network design problem heuristics, SIAM
  Journal on Algebraic Discrete Methods 1~(1) (1980) 51--63.

\bibitem{chuzhoy2008approximability}
J.~Chuzhoy, A.~Gupta, J.~S. Naor, A.~Sinha, On the approximability of some
  network design problems, ACM Transactions on Algorithms (TALG) 4~(2) (2008)
  23.

\bibitem{chan2006spanners}
T.-H.~H. Chan, M.~Dinitz, A.~Gupta, Spanners with slack, in: Algorithms--ESA
  2006, Springer, 2006, pp. 196--207.

\bibitem{Botton11}
Q.~Botton, B.~Fortz, L.~Gouveia, M.~Poss, Benders decomposition for the
  hop-constrained survivable network design problem, INFORMS Journal on
  Computing\href {http://dx.doi.org/10.1287/ijoc.1110.0472}
  {\path{doi:10.1287/ijoc.1110.0472}}.

\bibitem{Gouveia08}
L.~Gouveia, P.~Patricio, A.~Sousa, Hop-constrained node survivable network
  design: An application to mpls over wdm, Networks and Spatial Economics 8~(1)
  (2008) 3--21.

\bibitem{Pirkul03}
H.~Pirkul, S.~Soni, New formulations and solution procedures for the hop
  constrained network design problem, European Journal of Operational Research
  148~(1) (2003) 126--140.

\bibitem{veremyev_cnp_2015}
A.~Veremyev, O.~A. Prokopyev, E.~L. Pasiliao,
  \href{http://dx.doi.org/10.1002/net.21622}{Critical nodes for distance-based
  connectivity and related problems in graphs}, Networks 66~(3) (2015)
  170--195.
\newblock \href {http://dx.doi.org/10.1002/net.21622}
  {\path{doi:10.1002/net.21622}}.
\newline\urlprefix\url{http://dx.doi.org/10.1002/net.21622}

\bibitem{garey1979computers}
M.~R. Garey, D.~S. Johnson, Computers and intractability: a guide to
  np-completeness (1979).

\bibitem{albert2002statistical}
R.~Albert, A.-L. Barab{\'a}si, Statistical mechanics of complex networks,
  Reviews of modern physics 74~(1) (2002) 47.

\bibitem{bussieck2010minlp}
M.~R. Bussieck, S.~Vigerske, Minlp solver software, Wiley encyclopedia of
  operations research and management science.

\bibitem{van2007python}
G.~Van~Rossum, et~al., Python programming language., in: USENIX Annual
  Technical Conference, Vol.~41, 2007.

\bibitem{schult2008exploring}
D.~A. Schult, P.~Swart, Exploring network structure, dynamics, and function
  using networkx, in: Proceedings of the 7th Python in Science Conferences
  (SciPy 2008), Vol. 2008, 2008, pp. 11--16.

\bibitem{clark1990unit}
B.~N. Clark, C.~J. Colbourn, D.~S. Johnson, Unit disk graphs, Discrete
  mathematics 86~(1-3) (1990) 165--177.

\bibitem{zachary1977information}
W.~W. Zachary, An information flow model for conflict and fission in small
  groups, Journal of anthropological research (1977) 452--473.

\end{thebibliography}

\newpage
\appendix
\section{$NP$-Completeness of ECSTS Problem} \label{AppA}

We prove the $NP$-completeness of the ECSTS problem using a reduction from the \emph{Exact 3-Cover Problem}
~[\cite{garey1979computers}]. This problem is defined as follows:

\begin{definition}[Exact 3-Cover Problem] \label{ecp}
Let $T=\{T_1,T_2,\ldots ,T_{3t} \}$ be a set of $3t$ elements, for some integer $t$ and let $S=\{M_1,M_2,\ldots, M_k\}$ 
be a \emph{collection} of $3$-element subsets of $T$. The Exact $3$-Cover problem asks whether there exists a collection of 
subsets $S' \subseteq S$, of \emph{disjoint} $3$-element subsets of $T$, such that $\cup (M \in S') = T$.
\end{definition}

Mathematically we can write the problem as that of finding a subgraph $G_s =(V,E_s)$ of $G=(V,E)$ such that:

\begin{equation}
|E_s| \leq |V|-1 \ \textrm{ and} \ \mu_{G_s} \leq c, \ c \ \textrm{finite}
\end{equation} 

This in turn can be written as:

\begin{equation} 
|E_s| \leq |V|-1 \ \textrm{ and} \sum_{(u,v):u,v \in V;u \neq v} d(u,v) \leq C, \ C \ \textrm{finite}
\end{equation} 
where $c = \frac{C}{n(n-1)}$ where $n=|V|$ which is exactly the \ref{ecsts}
problem.

%\begin{figure}[!tbp]
\begin{figure}[ht]
  \centering
  \subfloat[]{\includegraphics[width=0.4\textwidth]{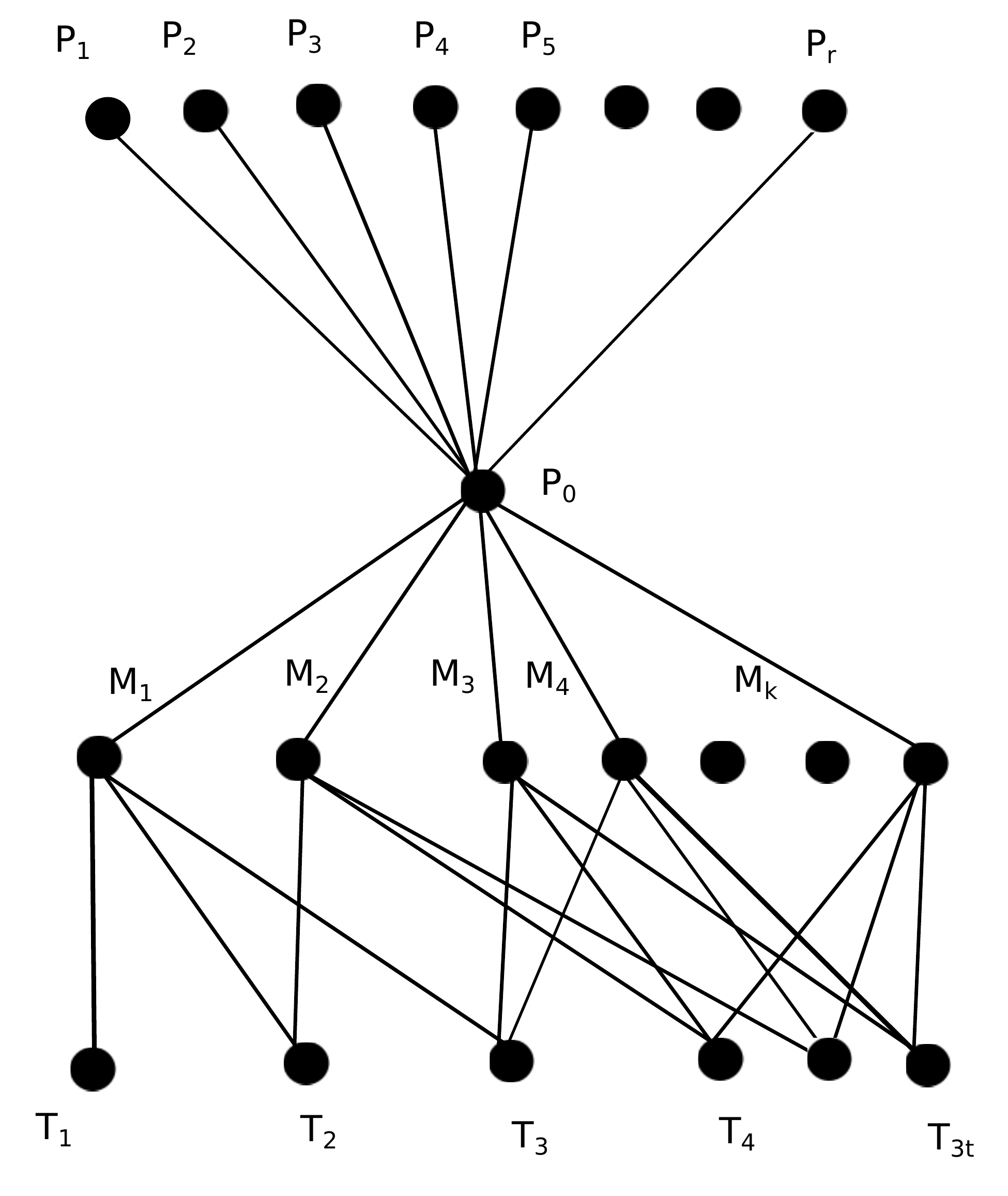}\label{fig:gadget}}
  \hfill
  \subfloat[]{\includegraphics[width=0.4\textwidth]{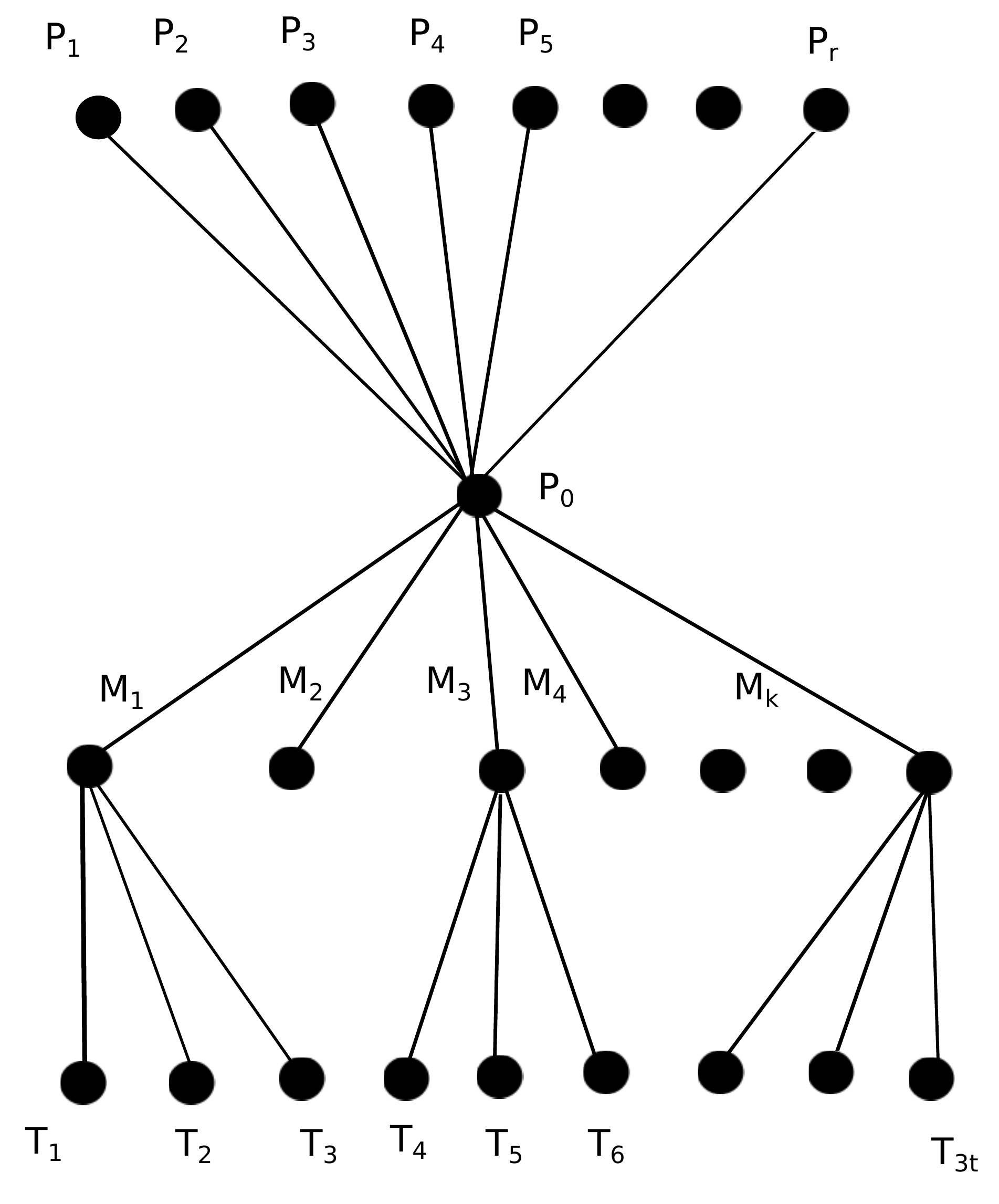}\label{fig:sp_gadget}}
  \caption{Illustration of the graphs used for the $NP$-completeness proof of the ECSTS problem (a) The graph obtained
  from the construction described in [\ref{cstr}] (b) The exact 3 cover solution}
\end{figure}

% Need to change the figures
% Will start after this point
The \emph{Exact 3-Cover} problem is known to be $NP$-complete~[\cite{garey1979computers}].
In order to show the $NP$-completeness of problem [\ref{ecsts}], we start with the construction of a graph. 
The vertex set of this graph is the union of three disjoint sets of vertices. We show that given a solution 
for problem [\ref{ecsts}] on this graph, we can get a solution to the exact 3-cover problem. This would complete 
the reduction.

\begin{definition}[Construction for Reduction] \label{cstr}
Let $T=\{T_1,T_2,\ldots ,T_{3t} \}$ be a set of $3t$ elements and let $S=\{M_1,M_2,\ldots, M_k\}$ be a collection
of $3$-element subsets of $T$. We construct a graph $G=(V,E)$. The vertex set $V = P \cup S \cup T$ where $P$ is a 
new set $\{ P_0,P_1,P_2, \ldots, P_r \}$ where $r$ is a constant, defined later. The value of $r$ depends on certain specific 
spanning trees of the graph $G$ as we will soon see. The edge set of the graph $G$ consists of the following:
\begin{itemize}
\item Edges $\{(P_0,P_i): i=1,2,3,\ldots,r \}$
\item Edges $\{(P_0,M_i): M_i \in S \}$
\item Edges $\{(M_i,T_j): M_i \in S, T_j \in M_i \}$
\end{itemize}
The resulting graph is shown in figure [\ref{fig:gadget}]. We also consider a particular type of spanning tree for the graph $G$.
This spanning tree is a feasible solution for the problem [\ref{ecsts}]. One such tree is shown in the figure [\ref{fig:sp_gadget}].
We are now ready to define the two constants $r$ as defined above and $C$ as defined in equation [\ref{ecsts}]. Let us denote 
by $\sigma_{SP}^G(A,B)$ the sum of the shortest paths between the vertices in the set $A$ and $B$ where $A,B \subseteq V \in G$,
for any given graph $G$. Let us denote the spanning tree in [\ref{fig:sp_gadget}] by $G_s$.  Now, we define the constants as follows:

\begin{itemize}
\item $r=\sigma_{SP}^{G_s}(S,S) + \sigma_{SP}^{G_s}(S,T) + \sigma_{SP}^{G_s}(T,T)$
\item $C = \sigma_{SP}^{G_s}(P,P) + \sigma_{SP}^{G_s}(P,S) + \sigma_{SP}^{G_s}(P,T) + r$
\end{itemize}

Thus $r$ is the sum of the shortest path distances between the vertices of $S$, the vertices of $S$ and $T$, and finally
the vertices of $T$ and $T$, in the spanning tree $G_s$. On the other hand, $C$ is the sum of the shortest path distances 
between the vertices of all these sets, both for intra-set vertex pairs as well as inter-set vertex pairs, computed with
respect to the spanning tree $G_s$. This completes the construction. We note that the graph $G_s$ as defined above is an 
Exact $3$-cover for the set $T$.
\end{definition}

Now we are ready to prove results that establish the $NP$-completeness of the problem. Thus we state the following result:

\begin{theorem} 
The ECSTS problem is $NP$-complete 
\end{theorem}

As stated before, in order to prove this theorem, we show that the Exact $3$-Cover problem has a solution, only if $G$ as 
defined above, contains a subgraph, that is a feasible solution for the problem [\ref{ecsts}]. We do this through the 
following lemmas.

\begin{lemma} \label{ecsts:l1}
Any spanning tree of the graph $G$ must contain the edges $\{(P_0,P_i): i=1,2,3,\ldots,r \}$.
\end{lemma}

\begin{proof}
The proof follows from the construction of the graph $G$. If, for example, the spanning tree does not contain the 
edge $(P_0,P_i)$ for some $i$, then there is no way to reach the vertex $P_i$ and thus the resulting graph is
disconnected and hence not a spanning tree. This contradicts our assumption. 
\end{proof}

\begin{lemma} \label{ecsts:l2}
Any spanning tree of the graph $G$, that is a feasible solution of the problem [\ref{ecsts}] must contain the edges 
$\{(P_0,M_i): M_i \in S \}$.
\end{lemma}

\begin{proof}
Let us suppose that there is a feasible solution $G_f$ such that it does not contain the edge $(P_0,M_i)$ for some $i$.
Now $\sigma_{SP}^{G_f} = \sigma_{SP}^{G_f}(P,P) + \sigma_{SP}^{G_f}(P,S) + \sigma_{SP}^{G_f}(P,T) + \sigma_{SP}^{G_f}(S,S) + \sigma_{SP}^{G_f}(S,T)
+ \sigma_{SP}^{G_f}(T,T)$. As each of these sums are non-negative, we have that 

\begin{align*}
\sigma_{SP}^{G_f} & > \sigma_{SP}^{G_f}(P,P) + \sigma_{SP}^{G_f}(P,S) + \sigma_{SP}^{G_f}(P,T) \\
      & > \sigma_{SP}^{G_s}(P,P) + \sigma_{SP}^{G_s}(P,S) + \sigma_{SP}^{G_s}(P,T) + 2 \cdot (r+1) \\
      & > C
\end{align*}
The second inequality follows from the fact that due to the absence of the edge $(P_0,M_i)$ for some $i$, the shortest 
paths between the vertices of the sets ($P$,$S$) and ($S$,$T$), changes. After some simple algebra, one can verify that
the second line of the inequality holds. Thus the subgraph $G_f$ cannot be a feasible solution for the problem [\ref{ecsts}] 
and hence the result follows. 
\end{proof}

\begin{lemma} \label{ecsts:l3}
Let $G_f$ be a feasible solution to the problem [\ref{ecsts}]. Then each vertex in $T$ is adjacent to exactly one vertex in
$S$.
\end{lemma}

\begin{proof}
First note that the assumption that $G_f$ is a feasible solution to problem [\ref{ecsts}] implies that $G_f$ is a spanning tree
for the graph $G$. We prove the result by contradiction. So let there be a vertex $T_j$ such that it is adjacent to two vertices 
$M_i$ and $M_k$ $\in S$.

We also have that $G_f \subseteq G$ where $G$ is defined in the construction [\ref{cstr}] above. Now by the same construction
and the lemma [\ref{ecsts:l2}], there are edges $(P_0,M_i)$ and $(P_0,M_k)$ in the graph $G_f$. Now this creates a problem
as $(T_j,M_i),(M_i,P_0),(P_0,M_k),(M_k,T_j)$ is a cycle and hence $G_f$ cannot be a spanning tree and hence it cannot be 
feasible for the problem [\ref{ecsts}] as we have assumed and this completes the proof. 
\end{proof}

\begin{lemma} \label{ecsts:l4}
Graph $G_s$ used in the construction [\ref{cstr}] above and depicted in the figure [\ref{fig:sp_gadget}] is a feasible solution 
of the problem [\ref{ecsts}].
\end{lemma}

\begin{proof}
The proof follows from lemmas [\ref{ecsts:l1}], [\ref{ecsts:l2}] and [\ref{ecsts:l3}].
\end{proof}

Now we are ready to prove the main theorem, which for convenience we state here again.

\begin{theorem} 
The ECSTS problem, as defined below, is $NP$-complete.
\begin{equation}
|E_s| \leq |V|-1 \ and \ \mu_{G_s} \leq C 
\end{equation} 
\end{theorem}

\begin{proof}
Let $G_f$ be any spanning tree for the graph $G$. As we have seen from the lemmas [\ref{ecsts:l1}], [\ref{ecsts:l2}] and [\ref{ecsts:l3}],
the spanning tree $G_f$ has a specific structure. Let us denote by $n_i$ the number of vertices in $S$ that are adjacent to exactly
$i$ vertices in $T$, $i = 0,1,2,3$.

\begin{align*}
\sigma_{SP}^{G_f}(T,T) & = 4(\frac{3t(3t - 1)}{2}) - 2 \cdot [(T_i,T_j): T_i \neq T_j, (M_k,T_i), (M_k,T_j) \in E_{G_f}, G \in S] \\
      & = (18 \cdot t^2 - 6 \cdot t) - (2 \cdot n_2 + 6 \cdot n_3) \\
      & = (18 \cdot t^2 - 12 \cdot t) + 6 \cdot t - 6 \cdot n_3 - 2 \cdot n_2 \\
      & = \sigma_{SP}^{G_s}(T,T) + 6 \cdot (t - n_3) - 2 \cdot n_2
\end{align*}
We note that $[ \cdot ]$ denotes the number of elements operator. From the above derivation we note that 
$\sigma_{SP}^{G_f}(T,T) = \sigma_{SP}^{G_s}(T,T)$ if and only if $n_3 = t$ and $n_i = 0, i=1,2$ and
hence $n_0 = s-t$. Now the condition $\sigma_{SP}^{G_f}(T,T) = \sigma_{SP}^{G_s}(T,T)$, by definition of $C$ is equivalent to the
fact that $\sigma_{SP}^{G_s} \leq C$ and hence feasibility of the [\ref{ecsts}] problem and the condition $n_3 = t$ and $n_i = 0, i=1,2$ 
and $n_0 = s-t$ is equivalent to the existence of an Exact $3$-cover. This completes the reduction. 
\end{proof}

%\newpage

\section{Lemmas On Greedy Algorithms} \label{AppB}
This appendix states and proves some of the properties of the solutions returned by the greedy algorithms. 
%For convenience we start off by stating the greedy algorithms again.

%\begin{algorithm}
%\DontPrintSemicolon % Some LaTeX compilers require you to use \dontprintsemicolon instead
%\SetKwFunction{algo}{GreedySpanner}
%\SetKwProg{myalg}{Procedure}{}{}
%\myalg{\algo{$G$, $t$}}{
%\KwIn{$G=(V,E,W)$ the input graph, $t$ the stretch}
%\KwOut{$G'$: The MECS Spanner}
%GreedySpanner($G=(V,E,W)$, $t$){
%$G' \gets (V,\{ \})$\;
%Sort $E$ in non-decreasing order of weights\;
%\For{$e=(u,v) \in E$} {
%  $P(u,v) \gets$ Shortest path between $u$ and $v$ in $G'$\;
%  \bf{if} {$P(u,v) > t \cdot W(e)$} {\bf then} Add $e$ to $G'$\;
  %\If{$P(u,v) > t \cdot W(e)$} {
  %  Add $e$ to $G'$\;
  %}
%}
%\Return{$G'$}\;
%}
%\caption{The Greedy Spanner Algorithm}
%\label{algo:greedyspanner}
%\end{algorithm}

We also recall that the algorithm [\ref{algo:greedyspanner}] is the standard spanner algorithm and looks at the all
pairs shortest paths and tries to preserve them to a constant multiple. We also recall the following lemma, which
was proved in the paper:

\begin{lemma}
Algorithm GreedySpanner always gives a feasible solution to the MECS problem.
\end{lemma}

The algorithms [\ref{algo:greedymecs}] and [\ref{algo:greedymecsv2}] are designed specifically for the MECS problem
and instead of considering the all pairs shortest paths individually, they consider the APL for making their greedy
choice.

%\begin{algorithm}
%\DontPrintSemicolon % Some LaTeX compilers require you to use \dontprintsemicolon instead
%\SetKwFunction{algo}{GreedyMECS}
%\SetKwProg{myalg}{Procedure}{}{}
%\myalg{\algo{$G$, $t$}}{
%\KwIn{$G=(V,E,W)$ the input graph, $t$ the stretch}
%\KwOut{$G$: The MECS Spanner}
%GreedySpanner($G=(V,E,W)$, $t$){
%$\mu \gets$ average distance in graph $G$\;
%Sort $E$ in non-increasing order of weights\;
%\For{$e \in E$} {
%  $G_{Temp} \gets G-{e}$\;
%  $\mu_{Temp} \gets$ average distance of $G_{Temp}$\;
%  {\bf if} {$\mu_{temp} > t \cdot \mu$} {\bf then} Do not remove $e$ from $G$ {\bf else} $G \gets G-{e}$\;
  %\uIf{$\mu_{temp} > t \cdot \mu$} {
  %  Do not remove $e$ from $G$\;
  %}
  %\Else{
  %  $G \gets G-{e}$\;
  %}
%}
%\Return{$G$}\;
%}
%\caption{The Greedy Removal MECS Algorithm}
%\label{algo:greedymecs}
%\end{algorithm}

%\begin{algorithm}
%\DontPrintSemicolon % Some LaTeX compilers require you to use \dontprintsemicolon instead
%\SetKwFunction{algo}{GreedyMECSV2}
%\SetKwProg{myalg}{Procedure}{}{}
%\myalg{\algo{$G$, $t$}}{
%\KwIn{$G=(V,E,W)$ the input graph, $t$ the stretch}
%\KwOut{$G$: The MECS Spanner}
%GreedySpanner($G=(V,E,W)$, $t$){
%$G' \gets (V,\{ \})$\;
%Sort $E$ in non-decreasing order of weights\;
%$\mu \gets$ average distance in graph $G$\;
%\For{$e \in E$} {
%  $\mu_{G'} \gets$ average distance of $G'$\;
  %$\mu_{Temp} \gets$ average distance of $G_{Temp}$\;
%  {\bf if} {$\mu_{G'} > t \cdot \mu$} {\bf then} $G' \gets G' \cup {e}$\;
  %\If{$\mu_{G'} > t \cdot \mu$} {
  %  $G' \gets G' \cup {e}$\;
  %}
%}
%\Return{$G'$}\;
%}
%\caption{The Greedy Addition MECS Algorithm}
%\label{algo:greedymecsv2}
%\end{algorithm}

In what follows we state as prove four lemmas that establish the claims made in Theorem 3 in the main paper. We state each claim
as a lemma in order that its easier to understand and the proof is not too large.

\begin{lemma} 
Algorithm [\ref{algo:greedymecs}]  always returns a feasible solution to the MECS problem.
\end{lemma}

\begin{proof}
Let us suppose the contrary. Suppose that the algorithm [\ref{algo:greedymecs}] returns a graph $G'$ that is not a
feasible solution for the MECS problem. Then $\mu_{G'} > stretch \cdot \mu_G$. This means that at some point in the
execution of the algorithm [\ref{algo:greedymecs}] an edge was removed from the input graph $G$ that resulted in the average of the
resulting graph to violate the MECS constraint. But this is not possible because the algorithm checks the condition to
make sure that this never happens. Thus we have a contradiction. Hence the result follows. 
\end{proof}

\begin{lemma} 
Algorithm [\ref{algo:greedymecsv2}]  always gives a feasible solution to the MECS problem.
\end{lemma}

\begin{proof}
The proof follows due to the fact that the for loop over the edges will continue to add edges to the graph
$G'$ as long as the constraint in the MECS problem is not satisfied. We also note that there is at least one feasible
solution, the graph $G$ itself and hence when the algorithm terminates it will return a feasible solution for
MECS. 
\end{proof}

Next we consider two lemmas that describe the structure of the solution returned by the algorithms [\ref{algo:greedymecs}]
and [\ref{algo:greedymecs}]. More precisely, we claim that the MECS solutions returned by these algorithms will always
contain the MST of the underlying graph as a subgraph.

\begin{lemma}
Consider the graph $G'$ returned by the algorithm [\ref{algo:greedymecs}] . For any input graph $G$, the graph $G'$
contains the MST of $G$. 
\end{lemma}

\begin{proof}
Let us suppose the contrary, that is let us suppose that the graph $G'$ returned by the algorithm does not contain the MST.
Now the algorithm [\ref{algo:greedymecs}] works by removing edges, starting with the heaviest edge first. Thus if the resulting
solution does not contain the MST, then it must be the case that at some iteration of the algorithm one of the MST edges are
removed. Let us consider the instant at which an MST edge $e=(u,v)$ is removed by algorithm [\ref{algo:greedymecs}]. The removal 
of this edge does not disconnect that graph and hence there is a path $p(u,v)$ between the vertices $u$ and $v$ of the graph.
Thus the edge $e$ that was removed is part of a cycle. Moreover, all the edges on the path $p(u,v)$ have weights that are less than 
the edge $e$. For otherwise, one of these edges would have been removed earlier and so removal of $e$ would have left the graph
disconnected.

Now let us consider the working of Kruskal's algorithm for computing the MST. It would start with the forest of all the vertices
of the graph and start adding edges in increasing order of their weights in the process merging two connected components. 
This in turn means that Kruskal's algorithm would create the path $p(u,v)$ between the vertices $u$ and $v$ before considering the
edge $e=(u,v)$ for addition. At this point, it would not add $e$ to the graph because $p(u,v) + e$ would form a cycle. Thus the 
MST resulting from a run of Kruskal's algorithm on $G$, would not contain the edge $e$.

This argument holds for each of the edges removed from the graph $G$ by algorithm [\ref{algo:greedymecs}]. Thus the solution returned 
by the algorithm, namely $G'$ will contain the MST of $G$ as a subgraph. This completes the proof. 
\end{proof}

\begin{lemma}
Consider the graph $G'$ returned by the algorithm [\ref{algo:greedymecsv2}] . For any input graph $G$, the graph $G'$
contains the MST of $G$. 
\end{lemma}

\begin{proof}
In order to prove this we note that the execution of the algorithm [\ref{algo:greedymecsv2}]  is very similar to the
Kruskal's Algorithm. Both of them start with a forest and then go on adding edges at each step, which results in
different components being merged together. We recall that in case of Kruskal's algorithm, an edge is added if and only
if it does not create a cycle. Let us denote by $G'_i$ $i=1,\ldots,Size(E)$ the collection of components generated by
algorithm~[\ref{algo:greedymecsv2}]. Similarly let us denote by $H_i$ $i=1,\ldots,Size(E)$ the collection of components generated
by Kruskal's Algorithm for the MST. We now prove that for each $i$ the number of connected components in $H_i$ is same
as the number of connected components in $G'_i$ and moreover each component of $H_i$ is the subset of a corresponding
component of $G'_i$. Once we have proved this our result is established.

We prove the result by induction on the number of edges. The base case is easy, $H_1$ contains a forest with each vertex
in its own connected component. The same also holds for $G'_1$. Then the base case is established. Let us assume that
the hypothesis is true for some $i$. Consider that the edge $e=(u,v)$ is considered for the $(i+1)^{th}$ step. There can
be two situations: $u$ and $v$ belong to the same component in $H_i$ or they are in different components of $H_i$.

\begin{description} 

\item[Same Component] 

If $u$ and $v$ are in the same component of $H_i$, then the edge $e$ will form a cycle. Hence Kruskal's algorithm will 
not add the edge and hence $H_{i+1}$ will not contain the edge $e$. Now in $G'_i$ the edge $e$ is in the same component 
by the inductive hypothesis. This component contains the component $H_i$. Now two things can happen, either the edge $e$ 
is not added to $G'_{i+1}$ in which case nothing changes and $H_{i+1}$ is still a subset of $G'_{i+1}$. Or else $e$ is added 
to $G'_{i+1}$ and then also the containment relationship does not change. Hence the result follows by induction.

\item[Different Components]

If $u$ and $v$ are in different components then the edge $e$ does not form a cycle. Thus
Kruskal's algorithm will add the edge and merge the two components of $H_i$ to get a new component of $H_{i+1}$. Now as
$u$ and $v$ belong to different components of $H_i$, by the induction hypothesis they also belong to different components
of $G'_i$. Thus there is no path between the two vertices $u$ and $v$ and hence the average distance will satisfy the
condition $\mu_{G'} > stretch * \mu$ in algorithm [\ref{algo:greedymecsv2}]. Thus the edge $e$ will be added and the 
corresponding components merged to form a single component in $G'_{i+1}$. Thus again we have that each component 
of $H_{i+1}$ is contained in a corresponding component of $G'_{i+1}$ and hence the result follows by induction. 

\end{description}
Now as $H_{Size(E)}=MST$ and $G'_{Size(E)}=G'$ our result follows. 
\end{proof}

Apart from being an interesting observation, with a nice and concise proof, the above result also helps us to change the
algorithm [\ref{algo:greedymecsv2}]. As we know that the solution returned by the algorithm will always contain the MST, 
instead of starting the algorithm from a forest, we can start the algorithm from the MST. Thus we compute the MST of the
input graph and then we start adding edges that are not included in the MST, in increasing order of weights. We continue
this process until the APL of the resulting graph satisfies the constraint.

\paragraph{Optimization of Algorithm [\ref{algo:greedymecsv2}]} 
We can do a simple optimization to the algorithm [\ref{algo:greedymecsv2}] in order to get rid of some redundancy in its 
operation. In order to do that we note that this algorithm considers the edges of the input graph $G$ in the order of edge
weights, in a non-decreasing order. Thus it may be the case that when an edge $e=(u,v)$ is being considered for addition to 
the graph $G$'by the algorithm, a path $P(u,v)$ already exists between the vertices $u$ and $v$. Thus if it is the case that 
$\mu_{G'} > t \cdot \mu_G$ and we consider adding the edge $e$ to the graph $G'$ as per algorithm [\ref{algo:greedymecsv2}],
then this this addition only makes sense if the following holds for the existing path between the vertices $u$ and $v$ in $G'$:
$P(u,v)_{G'} > t \cdot W(e)$, for if not then the path $P(u,v)_{G'}$ would satisfy $P(u,v)_{G'} \leq t \cdot P(u,v)_{G}$. We
know by lemma [\ref{l1}], that if we can ensure this for all pairs of vertices, then we have got a feasible solution to the 
MECS problem. This whenever $P(u,v)_{G'} \leq t \cdot P(u,v)_{G}$ and we are adding the edge $e=(u,v)$, we can safely ignore
adding the edge and still ensure the upper bound on the average path length. With this small optimization, we can expect to
reduce the weight of the resulting solution $G'$. With this change we can state the algorithm as follows:

\begin{algorithm}
\DontPrintSemicolon % Some LaTeX compilers require you to use \dontprintsemicolon instead
\SetKwFunction{algo}{GreedyMECSV2}
\SetKwProg{myalg}{Procedure}{}{}
\myalg{\algo{$G$, $t$}}{
\KwIn{$G=(V,E,W)$ the input graph, $t$ the stretch}
\KwOut{$G$: The MECS Spanner}
%GreedySpanner($G=(V,E,W)$, $t$){
$G' \gets (V,\{ \})$\;
Sort $E$ in non-decreasing order of weights\;
$\mu \gets$ average distance in graph $G$\;
\For{$e \in E$} {
  $\mu_{G'} \gets$ average distance of $G'$\;
  %$\mu_{Temp} \gets$ average distance of $G_{Temp}$\;
  \If{$\mu_{G'} > t \cdot \mu$} {
    $P(u,v) \gets$ Shortest path between $u$ and $v$ in $G'$\;
    \If{$P(u,v) > t \cdot W(e)$} {
      $G' \gets G' \cup {e}$\;
    }
  }
}
\Return{$G'$}\;
}
\caption{The Greedy Addition MECS Algorithm (Optimized)}
\label{algo:greedymecsv2opt}
\end{algorithm}

The output of algorithm [\ref{algo:greedymecsv2opt}] has several nice properties as well. We note that the change that
we have done to the algorithm [\ref{algo:greedymecsv2}] prohibits the formation of cycles and hence we can still claim
that the resulting solution contains the MST. Thus we can state that:

\begin{lemma}
Consider the graph $G'$ returned by the algorithm [\ref{algo:greedymecsv2opt}] . For any input graph $G$, the graph $G'$
contains the MST of $G$. 
\end{lemma}

The proof follows simply by the fact that we are breaking up cycles only and hence we are not changing the MST that is 
contained in the solution returned by the algorithm [\ref{algo:greedymecsv2}]. As a result we omit a detailed proof here. 

%\newpage

\section{Solutions for Unit Disk Graph} \label{AppC}

The following two figures show the results of running the exact MIP based solutions on the unweighted and weighted
unit disk graphs respectively. Figure [\ref{fig_uwd_mip}] shows the results for the unweighted graph and figure
[\ref{fig_wd_mip}] show the results for the weighted unit disk graph.

\begin{figure}[ht]
  \centering
\subfloat[$|E|=119$, $\mu=4.86$]{\label{fig_unit_disk_orig_1}\includegraphics[scale=0.1]{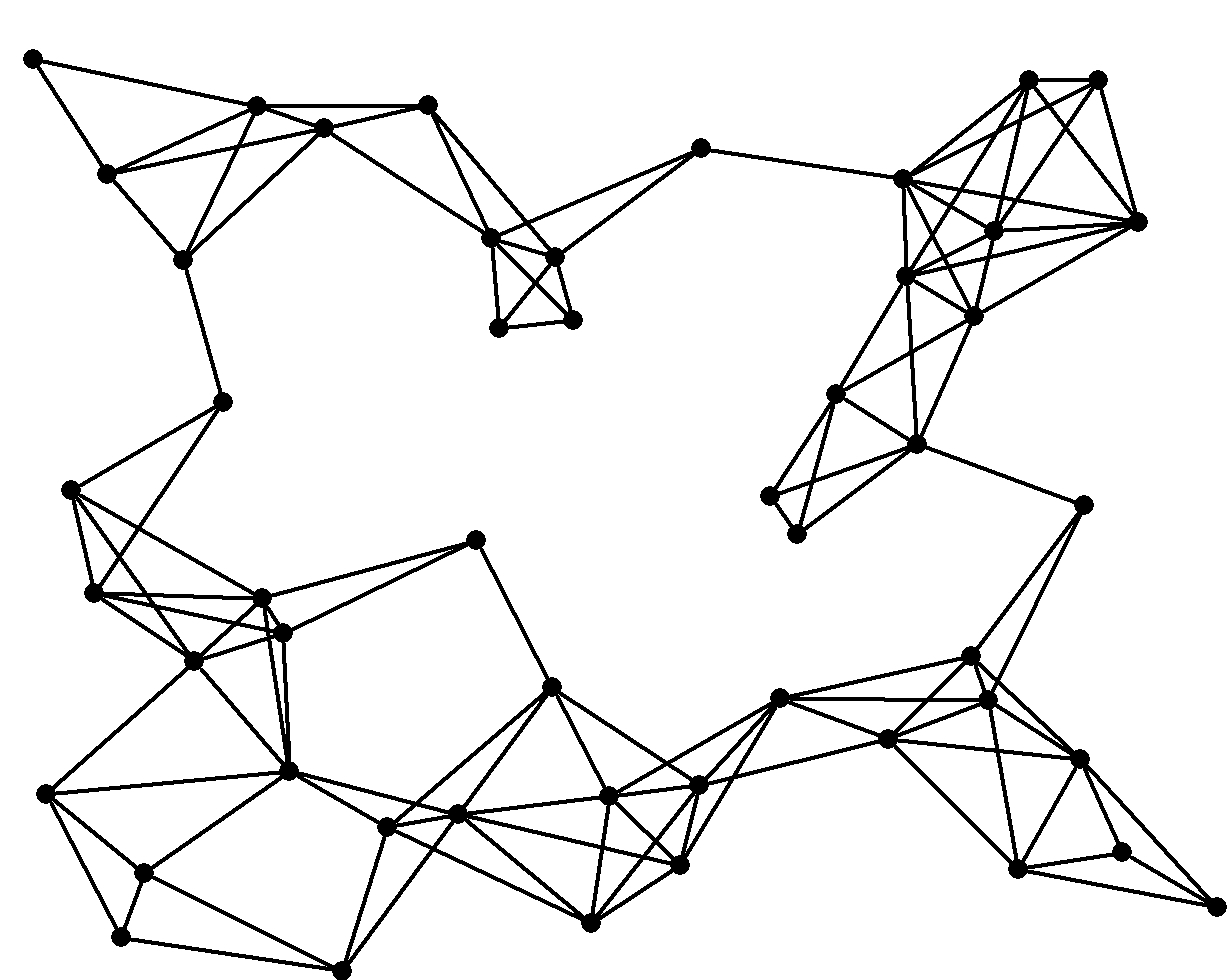}}
\subfloat[ $|E_s|=68$, $\mu_s=4.96$]{\label{fig_unit_disk_496_1}\includegraphics[scale=0.1]{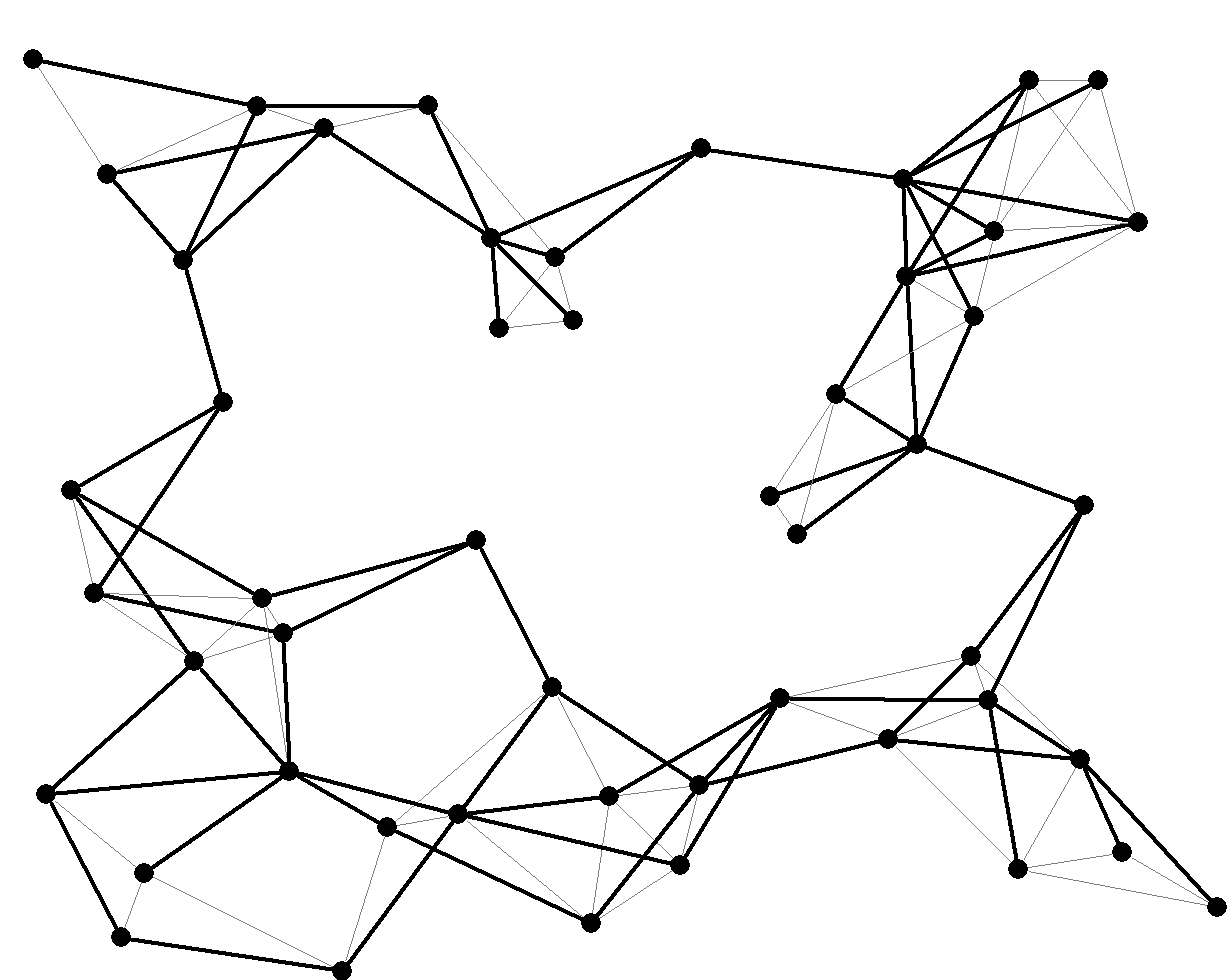}}\\
\subfloat[ $|E_s|=62$, $\mu_s=5.06$]{\label{fig_unit_disk_506_1}\includegraphics[scale=0.1]{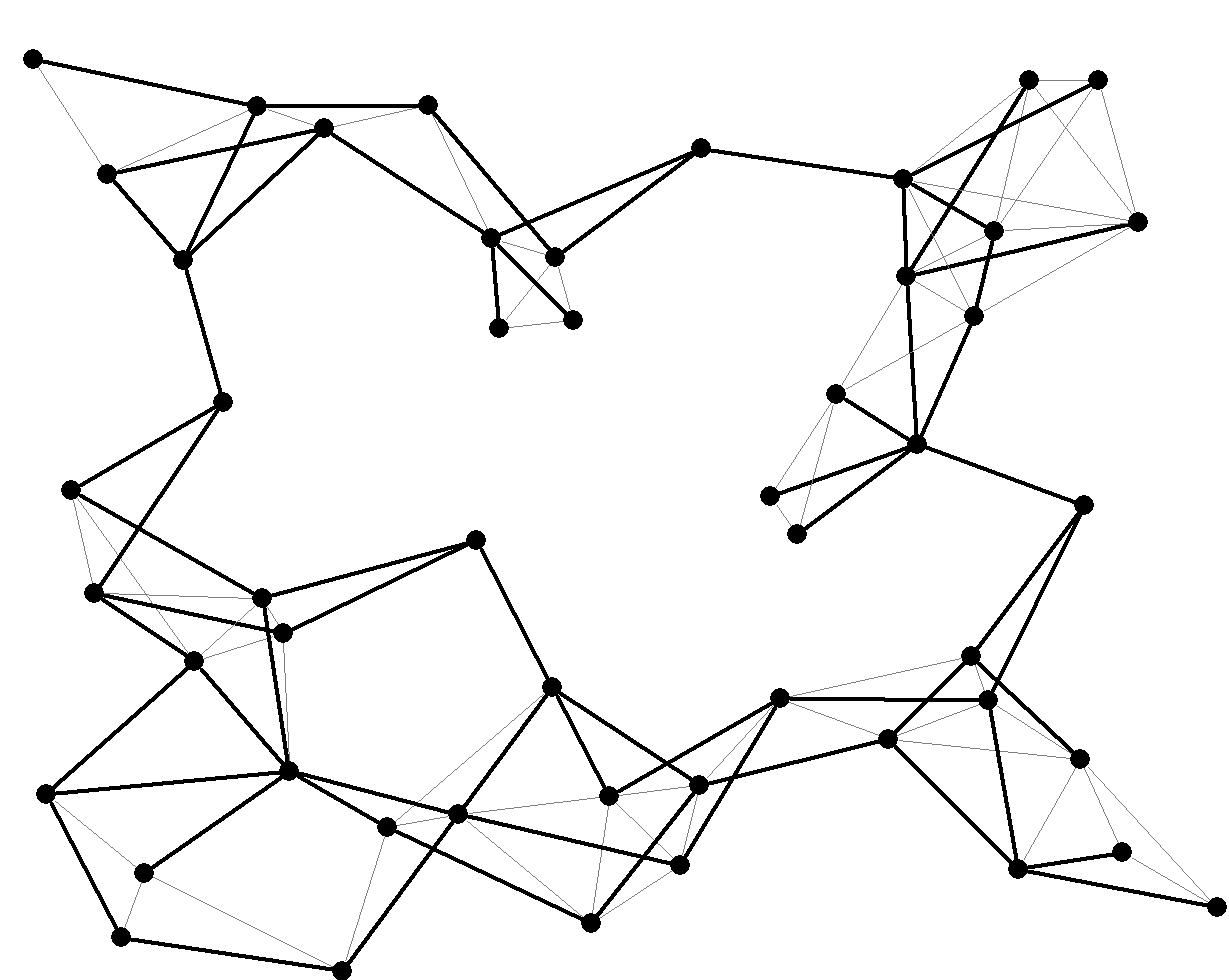}}
\subfloat[ $|E_s|=58$, $\mu_s=5.16$]{\label{fig_unit_disk_516_1}\includegraphics[scale=0.1]{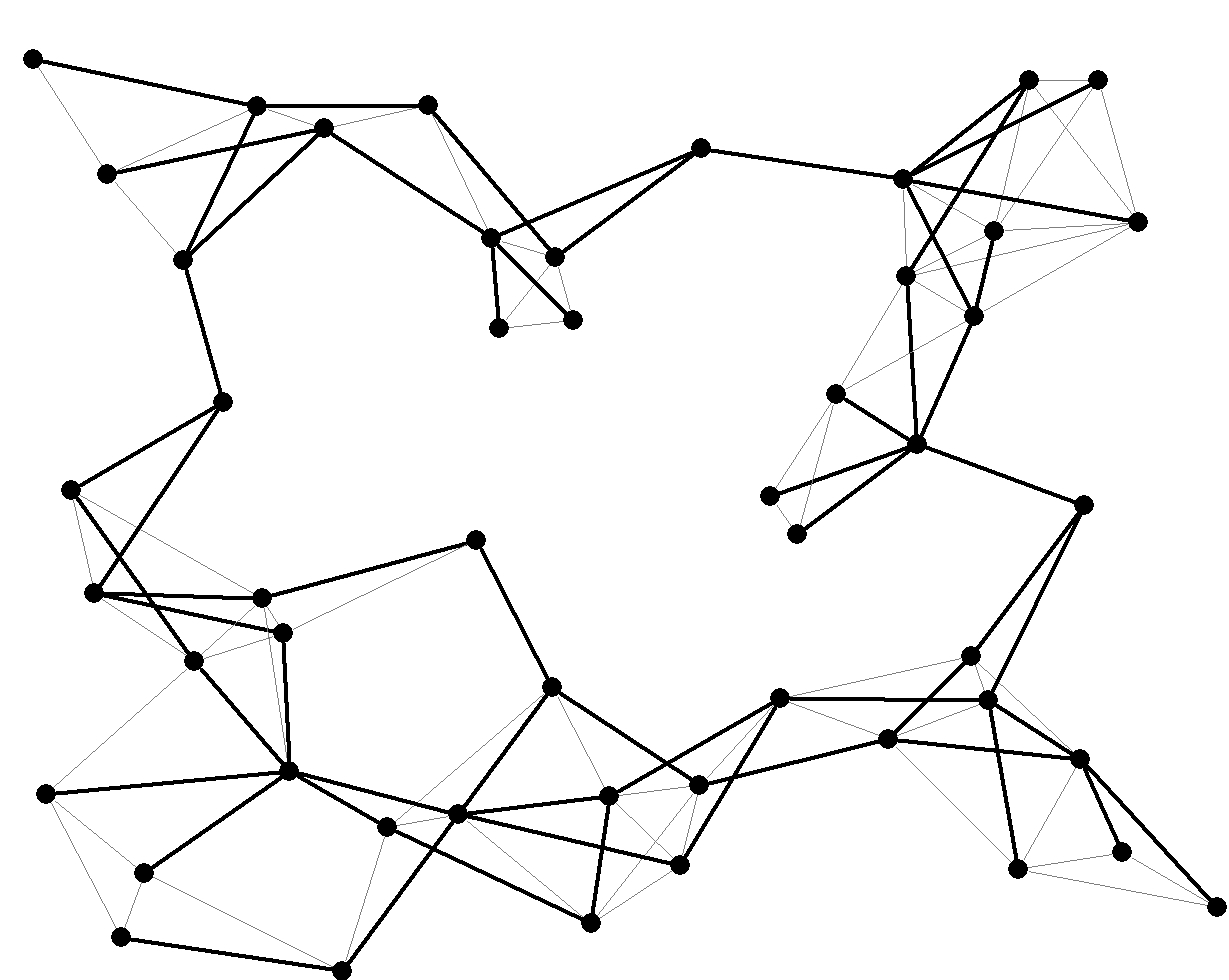}}
  \caption{Illustration of the \textbf{Unit Disk} randomly generated network instance ($|V|=50, |E|=119, D=11$) with average distance $\mu=4.86$ 
  and the optimal spanning subgraph (thick edges) with average distance (b) $\mu_s=\mu+0.1$, (c) $\mu_s=\mu+0.2$, (d) $\mu_s=\mu+0.3$.}
    \label{fig_uwd_mip}
\end{figure}

\begin{figure}[ht]
  \centering
\subfloat[$|E_w|=309$, $\mu=6.1$]{\label{fig_unit_disk_orig_2}\includegraphics[scale=0.1]{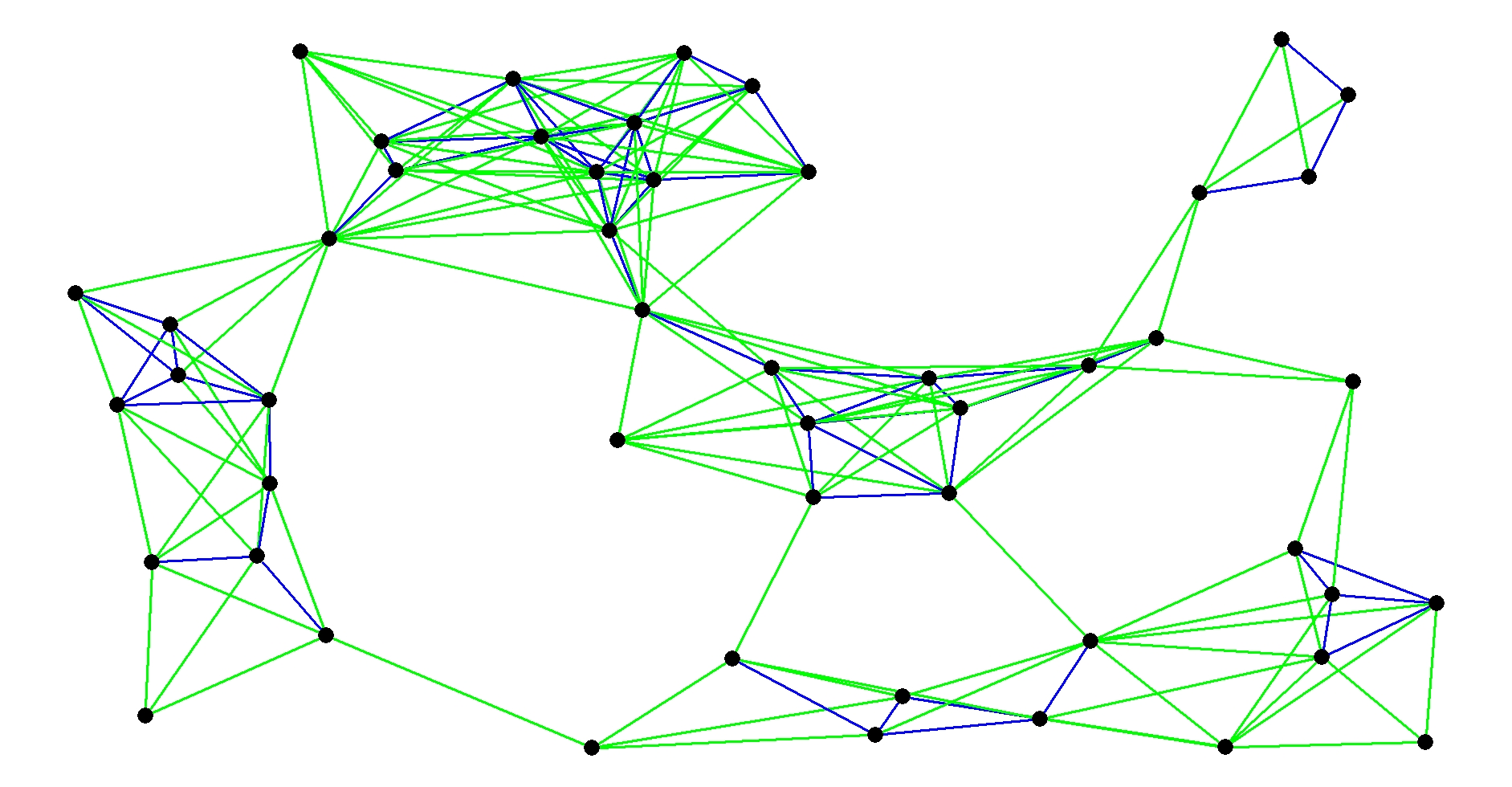}}
\subfloat[ $|E_s|=119$, $\mu_s=6.2$]{\label{fig_unit_disk_496_2}\includegraphics[scale=0.1]{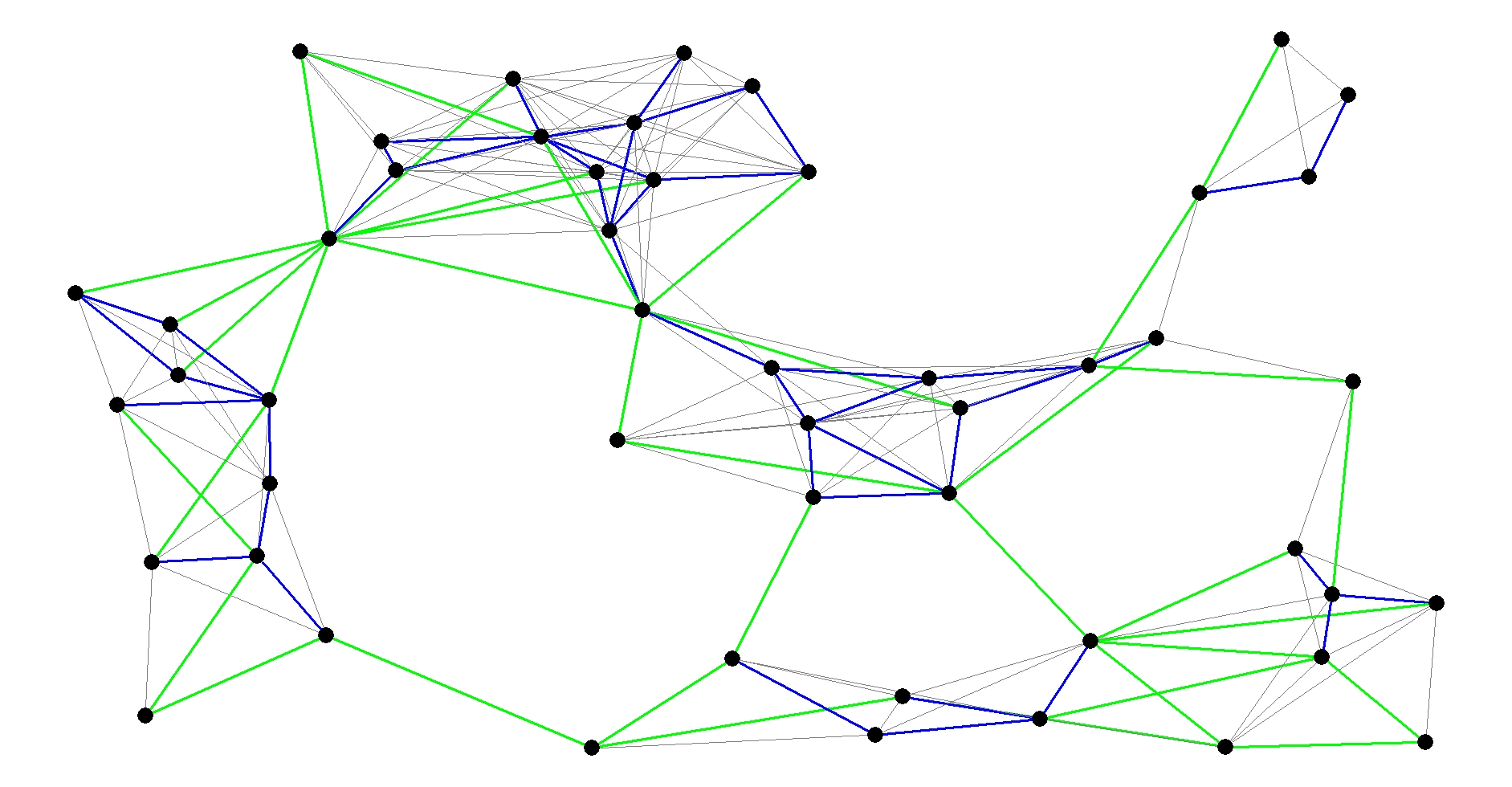}}\\
\subfloat[ $|E_s|=102$, $\mu_s=6.3$]{\label{fig_unit_disk_506_2}\includegraphics[scale=0.1]{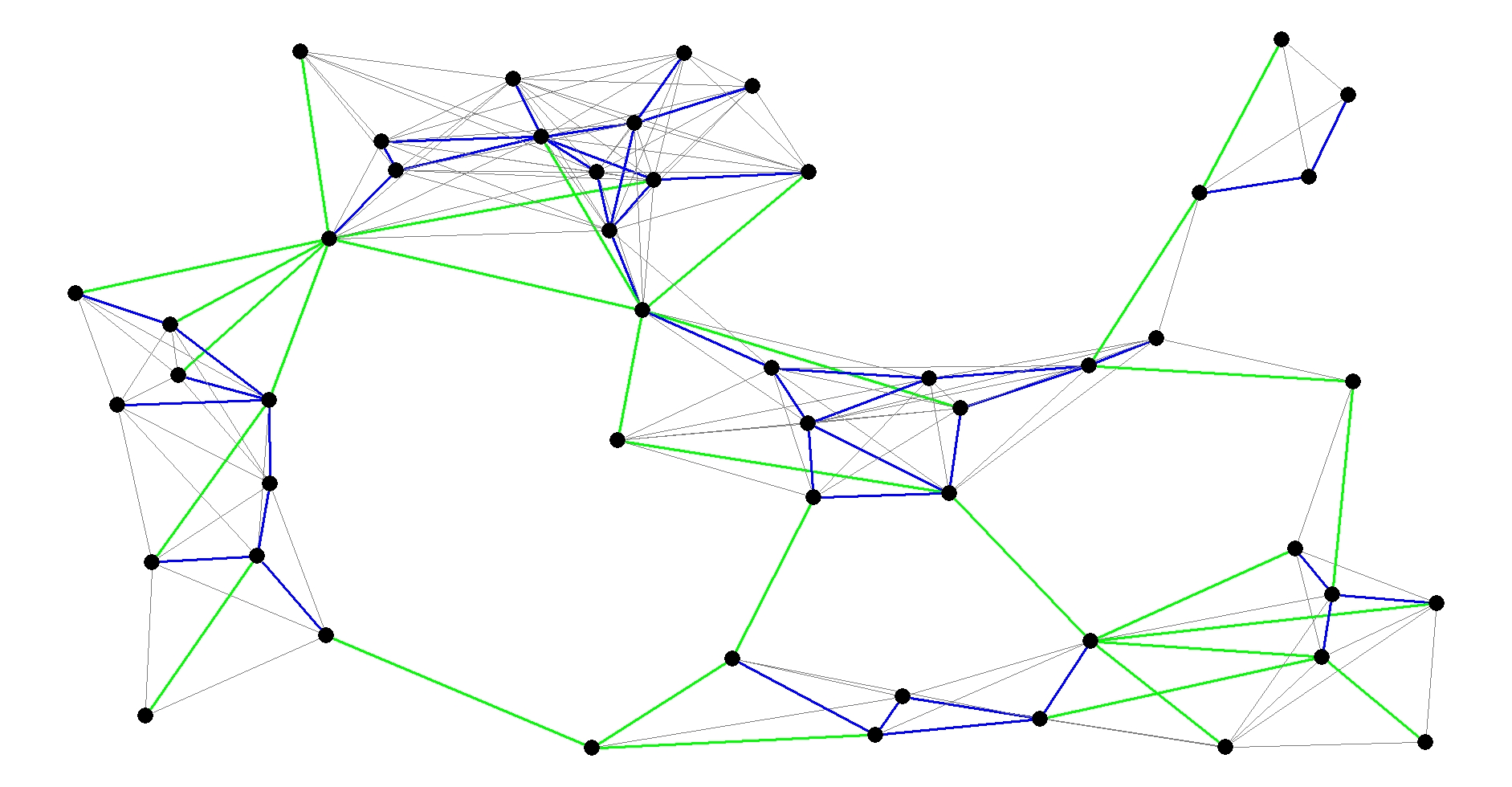}}
\subfloat[ $|E_s|=93$, $\mu_s=6.4$]{\label{fig_unit_disk_516_2}\includegraphics[scale=0.1]{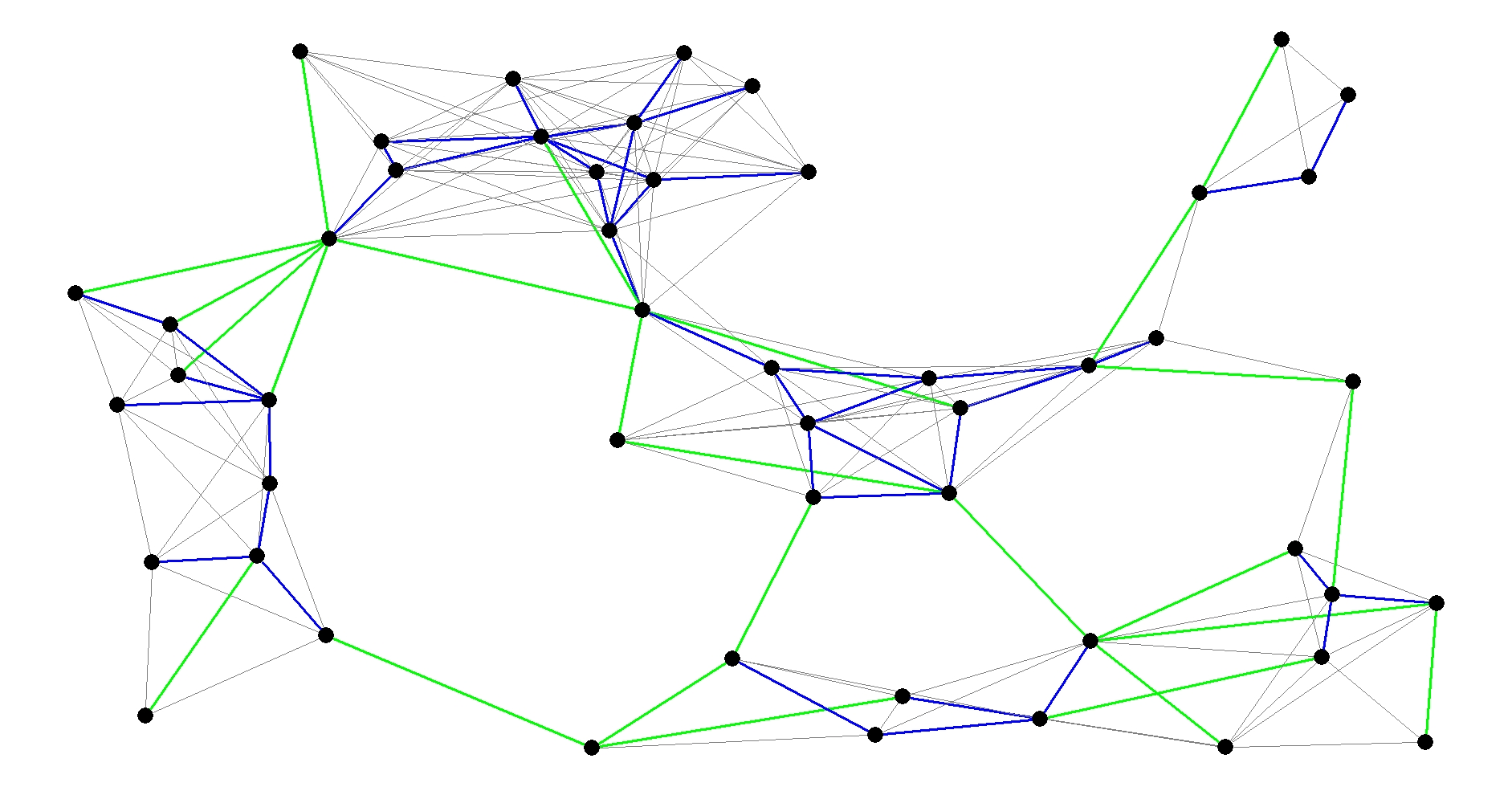}}
  \caption{Illustration of the weighted \textbf{Unit Disk} randomly generated network instance on a 100x100 plane ($|V|=50, |E_w|=309, D=15$) 
  with average distance $\mu=6.1$ and the optimal spanning subgraph (thick edges) with average distance (b) $\mu_s=\mu+0.1$, 
  (c) $\mu_s=\mu+0.2$, (d) $\mu_s=\mu+0.3$. We set $a_{ij}=1$ if euclidean distance between $i$ and $j$ is less than $12.5$ (blue edge), 
  and $a_{ij}=1$ if euclidean distance between $i$ and $j$ is less than $25$ and greater than 12.5 (green edge).}
  \label{fig_wd_mip}
\end{figure}

%\newpage
%\bibliographystyle{ijocv081} % outcomment this and next line in Case 1
%\bibliography{Alex,tm} % if more than one, comma separated
%\bibliography{Alex,tm}
% CASE 2: BiBTeX used to generate mypaper.bbl (to be further fine tuned)
%\input{mypaper.bbl} % outcomment this line in Case 2

%\end{document}

\end{document}